\documentclass[11pt]{article}

\usepackage{commath}
\usepackage{amsmath}
\usepackage{amsthm}
\usepackage{amssymb}
\usepackage{bbm}
\usepackage{booktabs}
\usepackage{color}
\usepackage{dcolumn} 
\usepackage{epigraph}
\usepackage{enumitem}
\usepackage{epstopdf}
\usepackage{graphicx}
\usepackage{epstopdf}
\usepackage{longtable} 
\usepackage[longnamesfirst]{natbib}
\usepackage{natbib}
\usepackage{rotating}
\usepackage{setspace} 
\usepackage{subcaption}
\usepackage{placeins}
\usepackage[margin=1.2in]{geometry}
\PassOptionsToPackage{hyphens}{url}\usepackage{hyperref} 
\hypersetup{
  colorlinks = TRUE,
  citecolor=blue,
  linkcolor=red,
  urlcolor=black
}

\newcommand{\starlanguage}{Significance indicators: $p \le 0.1:\ddagger$, $p \le 0.05:*$, $p \le 0.01:**$, and $p \le .001:***$.}

\DeclareMathOperator*{\argmax}{\arg\!\max}
\newtheorem{proposition}{Proposition}

\newcommand{\ET}{\text{\#EconTwitter}}
\newcommand{\RP}{\text{RePEc}}

\newif\ifcomments
\ifcomments
\newcommand{\OCD}{\clearpage \newpage}
\newcommand{\COMMENT}[1]{\textbf{\textcolor{red}{#1}}}
\else
\newcommand{\OCD}{}
\newcommand{\COMMENT}[1]{}
\fi

\newcommand{\numUsersNotclean}{68,592}
\newcommand{\numUsers}{55,966}
\newcommand{\medianFollowers}{2,004}
\newcommand{\medianFollowing}{913}
\newcommand{\medianRatio}{2.05}
\newcommand{\meanFollowers}{10,304}
\newcommand{\meanFollowing}{1,273}
\newcommand{\meanRatio}{13.98}

\newcommand{\ALLnumUsers}{26,284,663}
\newcommand{\ALLmedianFollowers}{266}
\newcommand{\ALLmedianFollowing}{524}
\newcommand{\ALLmedianRatio}{0.52}

\newcommand{\ALLmorethanoneRatio}{25.1}

\newcommand{\followingLower}{10}
\newcommand{\followerLower }{25}
\newcommand{\followerUpper }{200,000}
\newcommand{\followingUpper}{5,000}
\newcommand{\tweetLower    }{50}

\begin{document}  

\date{\today}

\newcommand{\WorkingTitle}{\textbf{The Production and Consumption of Social Media}}

\newif\ifdraft
\ifdraft
\title{\WorkingTitle}
\author{Anonymized for review}
\else 
\title{\WorkingTitle{}\thanks{
We thank
Hunt Allcott,
Erik Brynjolfsson,
Paul Oyer,
Evan Sadler,
Yo Shavit, 
Alex Tabarrok, and
Johan Ugander
for their very helpful comments.
We are grateful to Jie Lu for his excellent research assistance. 
Author contact information, code, and data are currently or will be available at \url{http://www.john-joseph-horton.com/}.
}} 
\author{Apostolos Filippas \\ Fordham 
\and 
		John J. Horton     \\ MIT \& NBER}
\fi

\newcommand{\AERAbstract}{
\noindent 
We model social media as collections of users producing and consuming content.
Users value consuming content, but doing so uses up their scarce attention, and hence they prefer content produced by more able users.
Users also value receiving attention, creating the incentive to attract an audience by producing valuable content, but also through \emph{attention bartering}---users agree to become each others' audience.
Attention bartering can profoundly affect the patterns of production and consumption on social media,
explains key features of social media behavior and platform decision-making, 
and yields sharp predictions that are consistent with data we collect from \ET{}.
}

\maketitle

\begin{abstract}
  \AERAbstract
\end{abstract} 

\newpage \clearpage
\onehalfspacing

\OCD
\section{Introduction}
Traditional media, such as newspapers, TV, and radio, are characterized by a few ``stars'' who produce content, and a much larger collection of people who purely consume that content~\citep{rosen1981economics, krueger2005economics}.
On social media, by contrast,  a large number of users both produce and consume content:
the ``social'' of ``social media'' is that millions of people
tweet on Twitter, 
dance on TikTok, 
rant on Facebook, 
open boxes on YouTube,
share sunsets on Instagram, 
and announce exciting new professional chapters on LinkedIn.
However, all of this new content has not been matched by an equal increase in the total supply of human attention---the scarce factor consumed during content consumption~\citep{simon1971designing}.
This paper examines how human attention is allocated in equilibrium over social media content---who consumes what---and how this allocation in turn affects the incentives for content production.

A necessary starting point is to observe that the costs of producing, distributing, and accessing content have fallen dramatically with  digitization, and with the proliferation of modern computing technology.
Today, marginal distribution costs are \emph{de minimis}, and technical production costs have diminished radically, even for audio and video content.
Falling costs can explain the larger number of professional producers, the greater variety of consumption, and why users spend more time on social media.
Nevertheless, producing content still has an opportunity cost.
Why then do so many amateur users produce content on social media, when almost no one is paid to do it?

Part of the explanation for amateur social media production is the desire for an audience and positive reaction from that audience.
Although people sometimes produce content without an external audience---diaries exist, and people enjoy playing music by themselves---many clearly value having an audience.
The desire for an audience can explain the effort put into social media---the production decision---but it can also affect the \emph{consumption} decision, in a way that is particular to social media.

We argue that the distinguishing feature of social media platforms is that they allow users to exchange some of their own attention for the attention of others, in order to obtain a larger audience.
Unlike most economic situations where individuals produce to consume, social media users can partially consume in order to produce. 
In particular, social media users can ``follow'' others not to consume their content, but to be ``followed back'' and consume their attention.
We call this exchange \emph{attention bartering}.

Attention bartering is an explicit part of the design of several social media platforms.
For example, Twitter allows unilateral following, and users can decide whether to engage in reciprocal following---``I'll follow you if you follow me.''
On other social media platforms where relationships are bilateral, such as Facebook and LinkedIn, attention bartering is informal and takes the form of users engaging with each others' content---``I like your posts, but only if you like mine.''   

Consider Alice, a user who decides where to allocate her attention.
On traditional media, Alice can only use her attention to consume content, and hence she allocates her attention towards consuming the best possible content---content produced by ``stars.''
Alice also likes receiving attention, but is not quite good enough to command an audience by producing content alone.\footnote{
We use ``ability'' here and throughout the paper as a measure of the consumption value other users place on the content without any implied judgement of the true value or worth of that content.
One of the most ``able'' twitter users by our definition was the $45$th President of the United States, who routinely shared to his more than 70 million Twitter followers ``content'' that would have otherwise been confined to poorly photocopied newsletters and email chains---before his suspension on January 8, 2021. {\color{white} :-)}
}
On social media, however, Alice can strike a ``deal'' with similarly good-but-not-great user Bob.
Alice agrees to give Bob some attention, in exchange for Bob giving some attention to her.
With this deal in place, they both get an audience member, and they now have the incentive to start producing content.
If Bob follows Alice then we say that Bob is Alice's follower, and that Alice is Bob's followee.

Both Alice and Bob would have not attracted any followers without the ability to attention barter, and hence they would have remained silent in a world without attention bartering.
The cost of their deal is the attention that they would have rather allocated elsewhere: if Alice was to renege from their deal and stopped following Bob, Bob would also renege because Alice's content is not good enough to attract his  attention ``organically.''
Their deal is sustained because social media platforms make it easy to verify that each holds their end of their bargain. 
Importantly, because attention bartering requires consuming what the partner produces, both parties have an incentive to seek out the best bargain, that is, the best producers who are willing to barter with them.
Together, these incentives create pressure for a kind of assortative matching.

We formalize the Alice and Bob situation in a model of attention bargaining, and solve for its equilibrium.
In our model, we assume that each user has limited attention and that there is substantial vertical differentiation among users.
We set aside horizontal differentiation to keep the analysis focused, but show later that it can be incorporated without qualitatively changing the results.
We also assume that attention from other users has increasing but diminishing returns.
As such, users are better off bartering with users who produce better content, all else equal, and two types of users do not attention barter in equilibrium:
(1) users with such low ability that others are not willing to barter with them, and 
(2) users with very high ability that attract followers organically, and  have no incentive to barter because attention has decreasing marginal returns. 
In-between these two extremes, attention bartering ``clubs'' emerge, whose ``members'' have similar  abilities, barter with all other members of the same club, and actively produce content.
Higher-ability barterers barter more, and select themselves into more populous clubs, comprising  members of higher average quality.
Note that clubs are not assumed explicitly, but emerge due to the production and consumption incentives of social media users.

Users in reciprocal following relationships---club members---would have not attracted any followers, and would not have produced any content without attention bartering.
With attention bartering, club members attract some followers and produce content actively, thereby experiencing lower consumption utilities, but higher attention utilities and higher total utilities from using the platform. 
As such, attention bartering results in more active users, and more content produced on the platform, albeit of lower average quality.

To assess some of the model assumptions and predictions empirically, we collect data from \ET{}, a Twitter community comprising professional and amateur users who tweet mostly about economics, and often follow each other.
An important feature of our data is that it allows us to observe the production and consumption decisions of \ET{} users, as well as to obtain measures of attention bartering and user ability. 

The distributions of \ET{} users' follower and followee statistics are consistent with our key assumptions of (a) vertical differentiation in ability, and (b) scarce attention. 
Follower counts are highly right-skewed, with some ``stars'' attracting enormous numbers of followers.
This is consistent with producers facing no distribution costs and being vertically differentiated, with the ``best'' users having vast audiences.
In contrast, the distribution of users' followees is substantially less right-skewed, and does not have a long tail.
This is consistent with consumers having limited attention---there are no ``super consumers'' that are able to consume 100x or 1000x the content of the typical user.
Despite this follower/followee distinction---and some star users having enormous ratios---follower-to-followee ratios are clustered around one.
These patterns are consistent with findings reported in previous work spanning several social media platforms, and over a long period of time~\citep{kwak2010twitter, myers2014information, sadri2018analysis}.
Many models of network formation can be made consistent with the patterns described above; what is distinguishing about our model are predictions about outcomes and behaviors conditional upon user ability.

We of course cannot obtain some objective measure of user ability, and trying to infer it from network statistics is backwards---we want to predict outcomes and behaviors from ability, not infer ability from outcomes.
As a proxy for ability, we use how many Twitter ``lists'' a user has been added to by other users, but adjusted for how active the account has been up until when we collected data.
Lists allow users to consume the content of listed members without affecting their follower and followee count.
Importantly, because lists are not salient within users' profiles, prospective followers cannot use list memberships to infer ability.
Because of these features of lists, our view is that list memberships are more likely to capture the actual consumption value users have for a given user's tweets, untainted by attention bartering considerations.

Consistent with our ability measure being a good proxy for what users want to consume, the average number of followers is strongly increasing in their ability.
In contrast to followers, the average number of followees is first increasing in ability, attains its maximum at medium-to-high ability levels, and then sharply decreases at the highest ability levels.
This ``followee dip'' is a telltale for attention bartering in our model---the most able users who attract many organic followers have little need to reciprocate, and so they only follow accounts they enjoy following.
It is of course possible that the most able users also happen to have more limited amounts of attention, but there is no obvious reason this would be the case.
Actually looking at the reciprocation decisions directly is illuminating on this point. 

In our data, we observe who follows whom, and we use reciprocal following as a proxy for attention bartering.
We find that low-ability users form reciprocal following relationships with just about any willing reciprocator, i.e., any user who is willing to ``follow back.''
As users' abilities increase they become more selective, following reciprocally higher-ability users on average.
Among all users, medium-ability users engage in reciprocal following most intensely.
In stark contrast to reciprocal following, no such pattern exists in unreciprocated following:
the content that users are willing to consume organically is unrelated to their own abilities, and is of high average quality.
The patterns are what our model predicts.

Although we offer a novel perspective on social media, the economic phenomena we highlight have social and economic antecedents. 
For example, consider an ``open mic night'' at a cafe.
The participants are not paid for the content they produce, but rather barter for attention: would-be poets listen to the poetry of their fellow poets, in exchange for their own poetry to be listened to. 
The cafe owner---like the operators of social media platforms---provides the meeting place for this exchange, profiting off the sale of food and drinks.
They are happy to sit silently in the background---insofar as the poets do not get too rowdy. 
In the context of academia, paper-reading seminars among grad students, and all-day workshops exhibit similar economics: academics listen to others' work for their work to be listened to.
Presenting at a workshop but not sticking around to hear others' talks is a \emph{faux pas} that only the most famous and high-ability academics can pull off safely.
In both the open mic night and the seminar/workshop, individuals would like to be members of the best club that will have them, as they would prefer to hear better-written poems, and to attend presentations of higher-quality papers.
However, the higher the quality of the other members, the less likely club membership is to be obtained. 

The social media platform assumes a passive role in our analysis, but the importance of attention bartering can also explain several platform strategy decisions.
Platforms have the incentive encourage attention bartering, because it increases the utilities of users who barter, inducing them to be more active on the platform, and in turn creating more opportunities to present ads.
Platforms can encourage attention bartering in several ways, including: 
rendering follower and followee content highly salient and creating reminders of the presence of the audience (such as ``likes''); 
allowing users to verify easily whether they are being followed by others; 
instituting rules against aggressive follow-churn behavior (which is a kind of attention bargaining fraud/defection); 
and implementing features to be used in lieu of unfollowing, such as ``muting'' and algorithmic curation.
Algorithmic curation of the feed spares the user some of the disutility of following low ability users, but without them giving up the attention of having that low ability user as a follower.
However, this is a line the platform has to walk carefully because algorithmic curation weakens the promise that any follower is in fact seeing the followed user's content, weakening the attention utility that motivates production.\footnote{Algorithmic curation is on a continuum with ``shadow banning''---a forum moderation technique designed to drive away users through reducing  nullifying engagement with their content, by means of making their content invisible to non-shadow banned users.}

One advantage of our descriptive model perspective is that it generates non-obvious insights for how policy changes might affect the equilibrium. 
For example, rather than ``banning'' users, a platform could simply prevent their following choices from being visible to others, thus prohibiting them from credibly bartering without silencing them.
The patterns of reciprocal following also makes clear the contours of ``clubs'' which in turn allows the platform to more easily remove troublesome users root and branch.
For example, the mass removal of ``Q anon'' conspiracy accounts on Twitter was surely facilitated by a high degree of reciprocal following, and was not purely determined by the content.\footnote{See \url{https://www.npr.org/2020/07/21/894014810/twitter-removes-thousands-of-qanon-accounts-promises-sweeping-ban-on-the-conspir}.}

This paper joins a burgeoning literature on social media~\citep{allcott2017social, alatas2019, allcott2020welfare,  levy2020social}, but is distinctive in taking a particularly economic lens.
Our modeling focus is on the importance of users' scarce attention in both their production and consumption decisions, complementing previous work examining the effects of algorithmic features~\citep{bakshy2015exposure, berman2020curation, levy2020social}.
Whereas we abstract away from the precise reasons why users want an audience, taking that desire as a given, several papers have explored the bases of this desire~\citep{toubia2013intrinsic, del2016echo, pennycook2021shifting}.

User connections on social media form graphs, and hence our model can be viewed as an economic network formation model, with rational economic agents severing and creating links to maximize their utilities~\citep{jackson2010social}, albeit adapted to our particular application.
One important assumption of our model is that cooperation between would-be barterers is sustainable, due to the repeated nature of the underlying game, and because forming and disengaging from relationships is costless on social media.
As attention bartering seemingly has substantial explanatory power, cooperation could be a fruitfully embedded in network formation models---see \cite{overgoor2019choosing} and \cite{chandrasekhar2021network} for excellent overviews of recent work.

The rest of the paper is organized as follows. 
In Section~\ref{sec:how_it_works}, we provide a brief introduction to Twitter as a paradigmatic social medium that allows for attention bartering.
In Section~\ref{sec:model}, we present a model of attention bartering, derive the main results, and examine extensions.
In Section~\ref{sec:empirics}, we assess empirically some of the predictions of the model using data from Twitter.
We discuss the implications of our results in Section~\ref{sec:discussion}, and we conclude in Section~\ref{sec:conclusion}. 
 
\OCD
\section{How Twitter works} \label{sec:how_it_works}
Twitter is a platform that allows its users to generate tweets---snippets of text, URLs, images, and videos, which do not exceed 280 characters---and share them with other users.
Twitter serves as the main motivating example throughout the rest of this paper.
For brevity, we henceforth use ``Twitter'' in lieu of ``social media'' whenever it is  possible, but our results should be understood to apply to any social medium that shares the features of Twitter that enable attention bartering.

A user's tweets are publicly available, unless she has elected to ``protect'' them by making her account private, or by ``blocking'' another user.
Users choose to see the tweets of others by ``following'' them.
Upon logging in Twitter, a user's ``timeline'' displays a stream of tweets from her followees---the users she follows.
Twitter allows unilateral following: any user is free to follow any
other user.\footnote{Social media platforms Instagram, Youtube, and TikTok also allow unilateral following.
However, on some social media, such as Facebook and LinkedIn, all relationships are bilateral.}
In addition to the tweets of followees, other content may be displayed on a user's timeline, including suggestions of users to follow, tweets from users not followed, and advertisements.
Users can choose between seeing tweets displayed in chronological order, or using Twitter's proprietary algorithmically generated ranking.\footnote{Twitter's ``top tweets'' feature produces  algorithmically curated tweet rankings, utilizing a variety of signals such as how popular a tweet is, 
and how the user's followees are interacting with it.
For more details, see~\url{https://help.twitter.com/en/using-twitter/twitter-timeline}.} 

The number of a users' followers and followees is prominently displayed.
Figure~\ref{fig:screenshots_profile} shows a Groucho Marx Twitter parody account, which has about 10,800 followers and 734 followees.
Additional content the user may opt to provide is also reported, such as the user's location.
\begin{figure}[h!]
\begin{center}
\caption{Elements of Twitter's information design.} 
\label{fig:screenshots}
\begin{subfigure}[t]{0.7 \textwidth}
\caption{The profile page of a twitter user.}
\vspace{-16pt}
\label{fig:screenshots_profile}
\center
\fbox{\includegraphics[width = \linewidth]{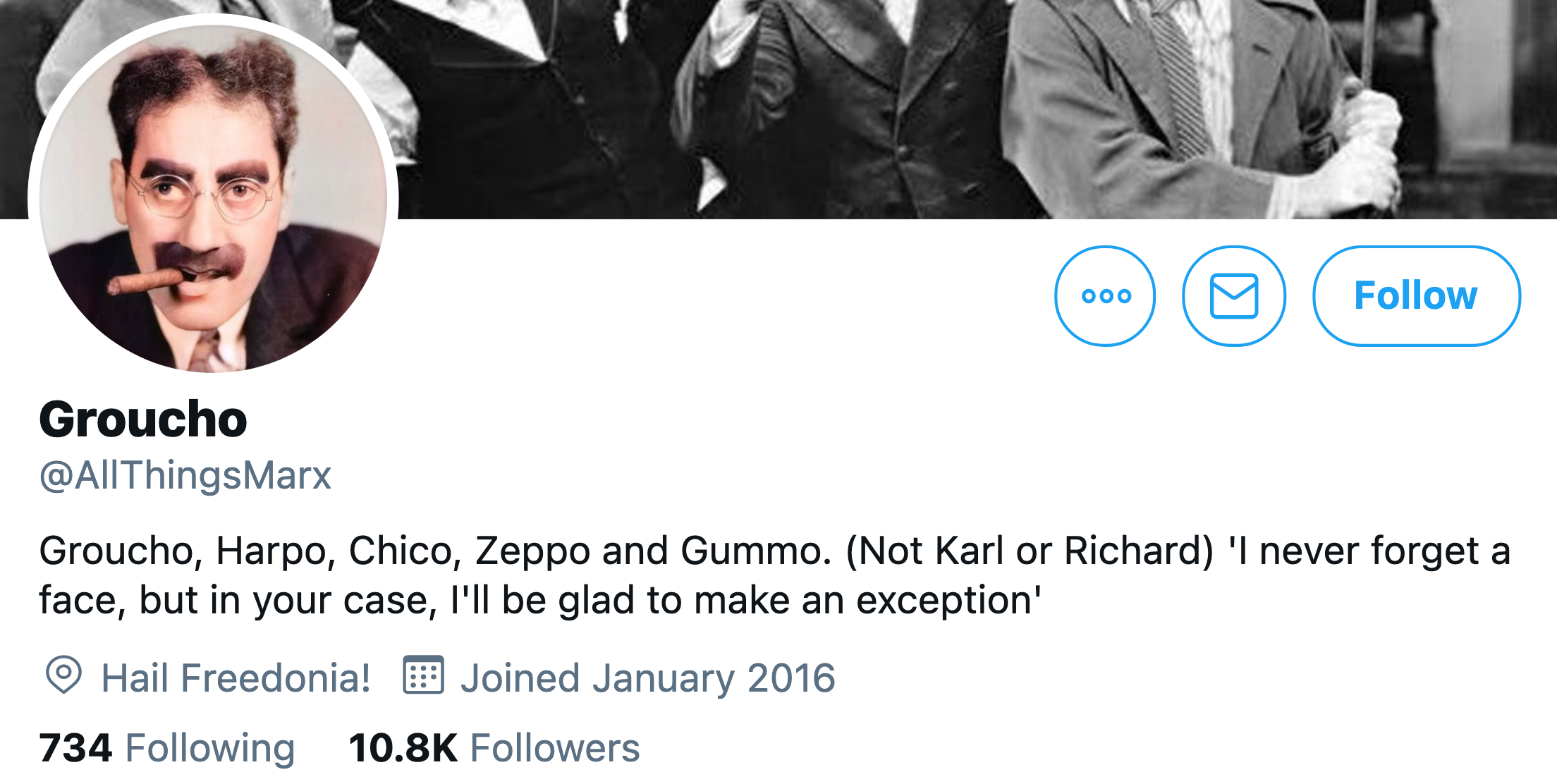}}
\vspace{10pt}
\end{subfigure}
\vspace{20pt}
\begin{subfigure}[t]{0.7 \textwidth}
\caption{Features of the relationship of two Twitter users.}
\vspace{-1pt}
\label{fig:screenshots_relationships}
\fbox{\includegraphics[width = \linewidth]{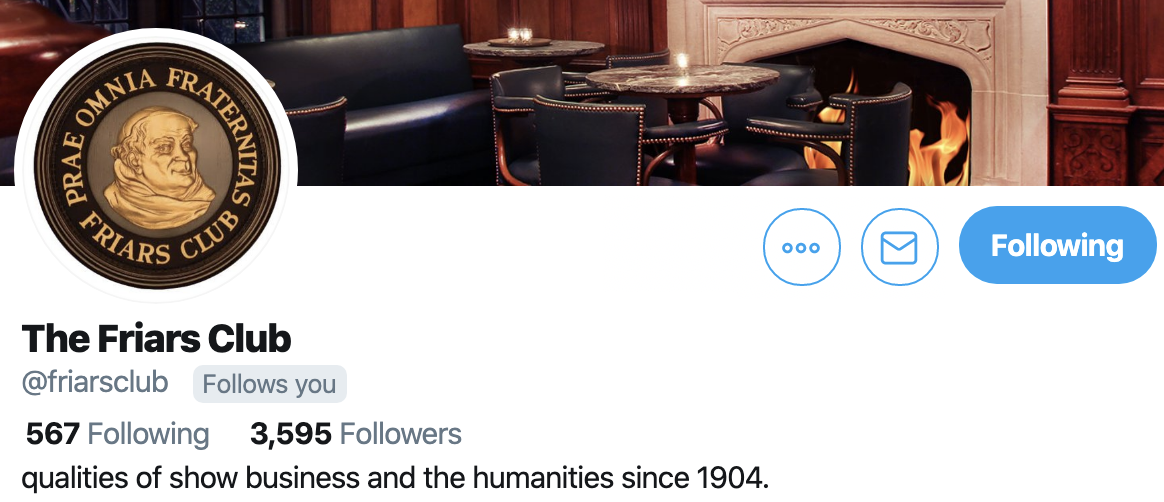}}
\end{subfigure}
\vspace{20pt}
\begin{subfigure}[t]{0.85 \textwidth}
\caption{Features of a tweet.}
\vspace{-3pt}
\label{fig:screenshots_tweet}
\fbox{\includegraphics[width = \linewidth]{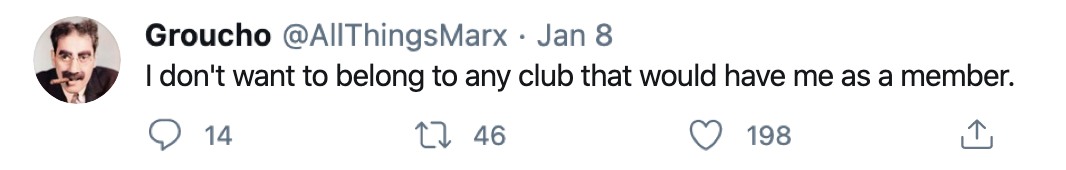}}
\end{subfigure}
\end{center} 
{\footnotesize \begin{minipage}{0.98 \linewidth}
\emph{Notes:}
This figure reports screenshots depicting core elements of Twitter's information design.
The top panel is a screenshot of the home page of a Groucho Marx parody Twitter account.
The numbers next to ``Following'' and ``Followers'' denote the number of followees and followers of this user.
The middle panel is a screenshot of the home page of The Friars Club Twitter account.
The page indicates prominently that ``The Friars Club'' follows the user, and that the user follows The Friars Club, creating a bidirectional following relationship.
Note that this relationship may have been either formed organically, or may be a result of the two users engaging in attention bartering.
The bottom panel is a screenshot of a tweet. 
The number of comments, retweets, and likes associated with that tweet, appear below the content of the tweet.
\end{minipage} }
\end{figure}
When a user views other users' profile pages, she can see whether she is following them, and whether they are following her.
Figure~\ref{fig:screenshots_relationships} shows an example of how other users' profiles are viewed when logged in as a user. 
In this example, the two users follow one another, indicated by the ``Following'' and ``Follows you'' tags.

A user can interact with any tweet that is visible to her by ``liking'' it, ``commenting'' on it by generating another tweet,  and ``retweeting'' it---which makes the original tweet appear on the  timeline of her followers.
The user who generated the original tweet receives notifications about the interaction other users have with her tweets. 
Figure~\ref{fig:screenshots_tweet} shows a tweet that has received 14 comments, 46 retweets, and 198 likes.

\OCD
\section{A model of social media production and consumption} 
\label{sec:model}

\subsection{Setup}
Consider a unit mass of Twitter users.
Each user is endowed with publicly known ability $\alpha \in [0,1]$, and generates tweets of quality $q(\alpha)$, where $q$ is some strictly increasing function.
For ease of exposition, we assume throughout the rest of the paper that $q(\alpha)= \alpha$, and that abilities are uniformly distributed on the unit interval.

Users obtain consumption utility by following others and consuming their tweets.
The consumption utility of a user is the sum of the tweet quality of her followees, minus an opportunity cost $q_0 \in [0,1]$ for each followee.
All users face the same opportunity cost to consumption.
With these assumptions, users have homogeneous consumption preferences. 

Users also obtain attention utility when their tweets are consumed by their followers.
Attention utility increases in the number of a user's followers, exhibits diminishing returns, and does not depend on the identity of the followers---only the number of followers matters.
Formally, attention utility is a function $I: [0,1]\rightarrow[0,1]$ such that $I' > 0$ and $I'' < 0$.

Users make an extensive margin decision---whether to tweet.
We assume that users generate tweets only if they can obtain positive attention utility; 
users without followers ``lurk,'' consuming but not generating tweets.
We assume that users have homogeneous attention preferences.
As such, the model assumes that no horizontal differentiation exists, and that users are identical except in their abilities.
We extend our model to allow for heterogeneous consumption and attention preferences in Section~\ref{ssec:heterogeneity}.

The coexistence of consumption and attention utilities gives rise to two distinct follower types: organic and reciprocal followers.
An organic follower is a user who follows some other user that increases her consumption utility.
Because preferences are homogeneous, only users with ability $\alpha > q_0$ have organic followers.
Users with ability $\alpha < q_0$ do not have organic followers, but may obtain followers by engaging in attention bartering.

A reciprocal follower is a user with ability $\alpha < q_0$  that follows a user with ability $\alpha' < q_0$,  only if that user follows her back.
In reciprocal following relationships, users engage in attention bartering,  trading off decreases in consumption utility for increases in attention utility.
In addition to the consumption utility decrease, we assume that users also incur a cost $c > 0$ for every reciprocal following relationship they engage in.
Because of the incentive to defect, users have to exert effort to maintain each reciprocal relationship---the parameter $c$ captures the cost of this effort.\footnote{For example, a user engaging in attention bartering may  periodically monitor that her reciprocal followers are holding their end of the bargain, and have not stopped following her back.}

We consider the equilibria under two platform configurations: 
(i) when attention bartering is impossible, which happens when follower information is not public and hence non-contractible, and
(ii) when attention bartering is possible, which happens when follower information is public.

\subsection{Equilibria}
Let $f_\alpha:[0,1] \rightarrow \{0,1\} $ denote the (follow) strategy of user $\alpha$, such that $f_\alpha(x) = 1$ if she follows users with ability $x$, and $f_\alpha(x) = 0$ otherwise. 
The set of strategies of a user is ${F}$, and the set of  strategy profiles is $\mathcal{F} = \times_{\alpha \in [0,1]} F$.
For a  profile  $f = (f_\alpha , f_{-\alpha}) \in \mathcal{F}$, we denote by $f_{-\alpha}$ the profile of all users with ability not equal to $\alpha$.
Note that because users have homogeneous preferences, we can refer to them by their abilities---in what follows, we refer to a user with ability $\alpha$  as ``user $\alpha$,'' whenever it is convenient.

Given profile $f$, user $\alpha$ 
(i) follows $n_\alpha(f)  = \int_0^1 f_\alpha(x) dx$ users, and obtains consumption utility  $U_\alpha(f) = \int_0^1 f_\alpha(x) \left( x - q_0 \right) dx$, 
(ii) is followed by $m_\alpha(f) = \int_0^1 f_x(\alpha) dx$ users, and obtains attention utility $I_\alpha(f) = I(m_\alpha(f))$, and
(iii) attempts to engage in $r_\alpha(f) = \int_{0}^{q_0} f_\alpha(x) dx$ reciprocal following relationships, and incurs monitoring cost $C_\alpha(f) = r_\alpha(f) c$.
User $\alpha$ obtains total utility equal to
\begin{equation}
	V_\alpha(f) = U_\alpha(f) + I_\alpha(f) - C_\alpha(f).
\end{equation} 

We restrict our attention to profiles any two users $\alpha, \alpha' \leq q_0$ will either engage in reciprocal following, or will not follow each another.
Formally, we consider profiles in the set  
\begin{equation}
\label{eq:reciprocal_following}
\mathcal{F}^+ = \{f \in \mathcal{F}: 
	f_{\alpha}(\alpha') = f_{\alpha'}(\alpha) \quad  \forall \alpha < q_0 \;\; \forall \alpha' < q_0\}.
\end{equation}
A profile $f^*\in \mathcal{F}^+$ is an equilibrium  profile if 
\begin{equation}
\label{eq:solution_concept}
    V_\alpha(f^*) \geq V_\alpha(f) \; \; \forall \alpha \;\; \forall f \in \mathcal{F}^+.
\end{equation}
It is worth noting that by restricting feasible profiles to $\mathcal{F}^+$, our equilibrium notion assumes a user who considers unfollowing a reciprocal follower takes into account that she will be ``unfollowed back.''
Without this restriction, no attention bartering would take place in equilibrium, as users are always better off unilaterally deviating from a reciprocal following relationship.
As such, we assume that cooperation is sustainable due to the ``repeated'' nature of the underlying game, and because forming and disengaging from following relationships is costless on social media.
Furthermore, note that Equation~\eqref{eq:reciprocal_following} requires that no coalition of users can block an equilibrium profile.

Consider a profile $f \in \mathcal{F}^+$ where users $\alpha, \alpha' \leq q_0$ reciprocally follow each other, and a  profile $f' \in \mathcal{F}^+$ where they do not, but is otherwise identical to $f$.
The change in the utility of user $\alpha$ from unfollowing reciprocal follower $\alpha'$ is:
\begin{align}
	V_\alpha(f') - V_\alpha(f) = 
	\underbrace{U_\alpha(f') - U_\alpha(f)}_{\mbox{{\footnotesize $\Delta U_\alpha > 0$}}} + 
	\underbrace{I_\alpha(f') - I_\alpha(f)}_{\mbox{{\footnotesize $\Delta I_\alpha < 0$}}} + 
	\underbrace{C_\alpha(f') - C_\alpha(f)}_{\mbox{{\footnotesize $\Delta C_\alpha > 0$}}}.
\end{align}
Disengaging from a reciprocal relationship leads to an increase in consumption utility, a decrease in attention utility, and a decrease in monitoring costs; engaging in a reciprocal relationship has the opposite effects.
It is this fundamental trade-off that our paper examines.

\subsection{Attention bartering is impossible} \label{sec:no_info}
We first examine the equilibrium outcomes in a platform where attention bartering is impossible. 
This can be thought of as the case of a social media platform where follower information is hidden, and hence non-contractible.
In terms of the parameters of our model, no attention bartering takes place when the monitoring cost $c$ is so large that reciprocal following always decreases a user's total utility.

When attention bartering is impossible, it is straightforward to show that every user follows only those users that increase her consumption utility, i.e., users with ability $\alpha>q_0$.
\begin{proposition} \label{prop:no_info_equilibrium}
In a platform without attention bartering, there exists a unique equilibrium $f^*$ such that for all $\alpha$, $f^*_\alpha(x) = 1$ for $x>q_0$, and $f^*_\alpha(x) = 0$ otherwise.
\end{proposition}

Figure~\ref{fig:no_info_uniform} illustrates the equilibrium of a platform where attention bartering is impossible, for the case of $I(x) = \sqrt{x}/2,$ $q_0 = 0.8$, and $c=0.2$.
Every user obtains  consumption utility 
\begin{equation}
	U_\alpha(f^*) = \int_{q_0}^{1}  x - q_0 \; dx = \frac{1}{2}(1-q_0)^2.
\end{equation}
Users with ability $\alpha \leq q_0$ lurk---they are not followed by anyone, and hence they do not generate tweets---and obtain total utility equal to their consumption utility, i.e., $V_\alpha(f^*)= U_\alpha(f^*)$.
Users with ability $\alpha > q_0$ are followed by every user, generate tweets,  and obtain  total utility $V_\alpha(f^*)= U_\alpha(f^*) + I(1)$.
\begin{figure}[h!]
\begin{center} 
\caption{Following, consumption, and attention with homogeneous user preferences.
\label{fig:individual_outcomes_unif}}
\vspace*{0.15in}
\begin{subfigure}[t]{1 \textwidth}
\caption{Equilibrium when attention bartering is impossible.}
\label{fig:no_info_uniform}
	\begin{subfigure}[t]{0.5 \textwidth}
  		\includegraphics[width = \linewidth]{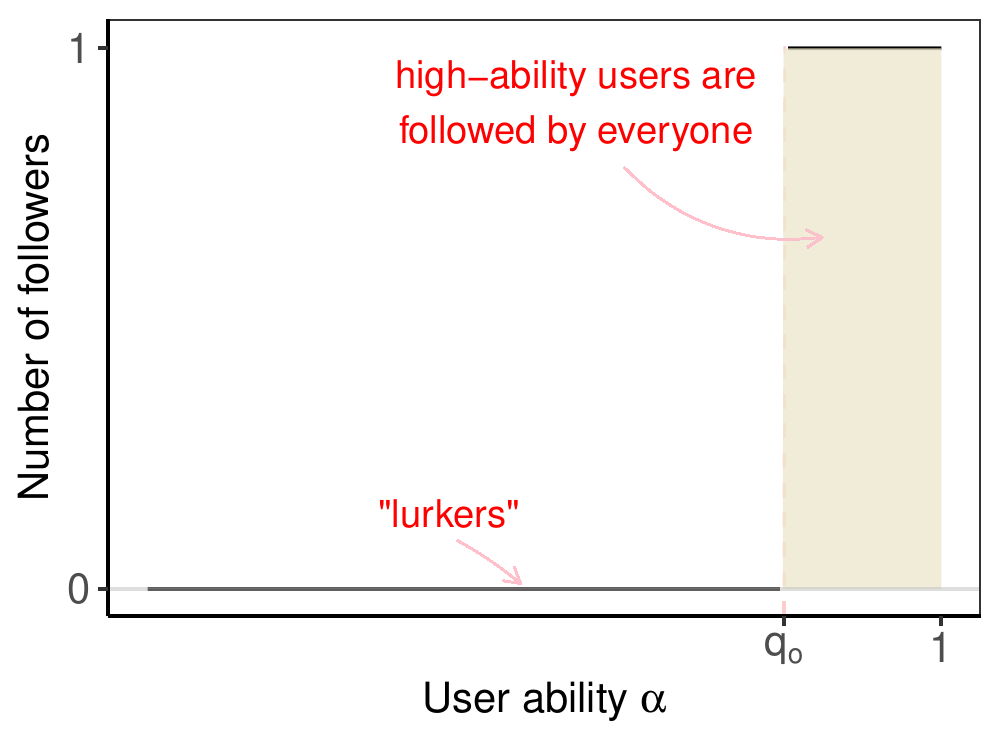}
		\label{fig:no_info_uniform_followers}
	\end{subfigure}
	\begin{subfigure}[t]{0.5 \textwidth}
  		\includegraphics[width = \linewidth]{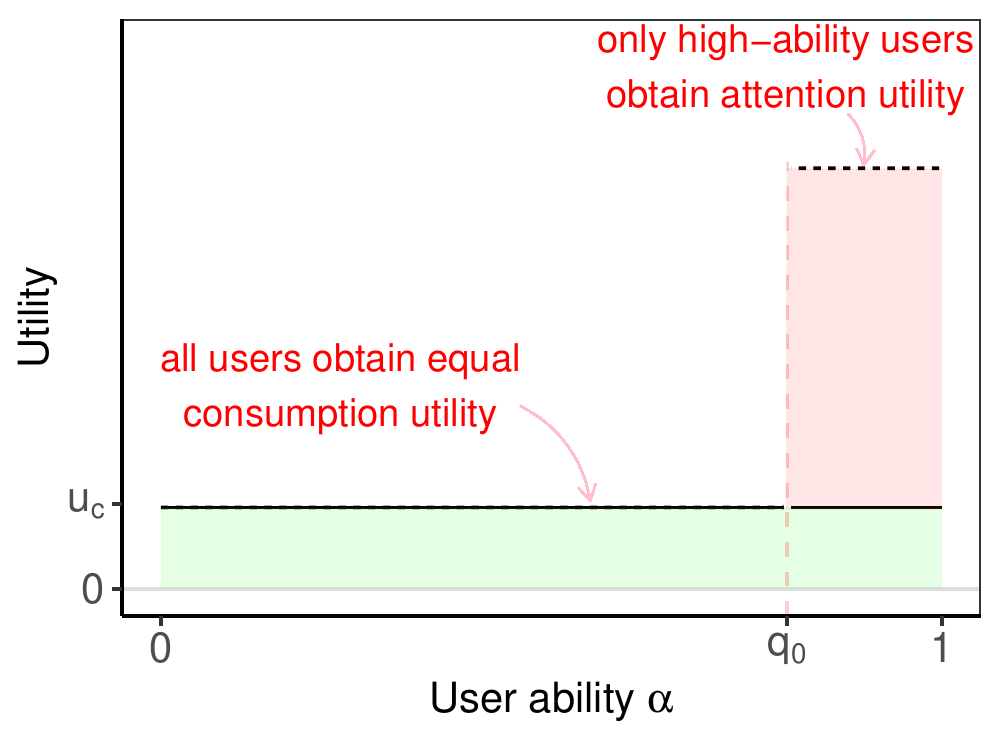}
  		\label{fig:no_info_uniform_utilities} 
	\end{subfigure}
\end{subfigure}

\vspace*{0.15in}
\begin{subfigure}[t]{1 \textwidth}
\caption{Equilibrium when attention bartering is possible.}
\label{fig:info_uniform}
	\begin{subfigure}[t]{0.5 \textwidth}
  		\includegraphics[width = \linewidth]{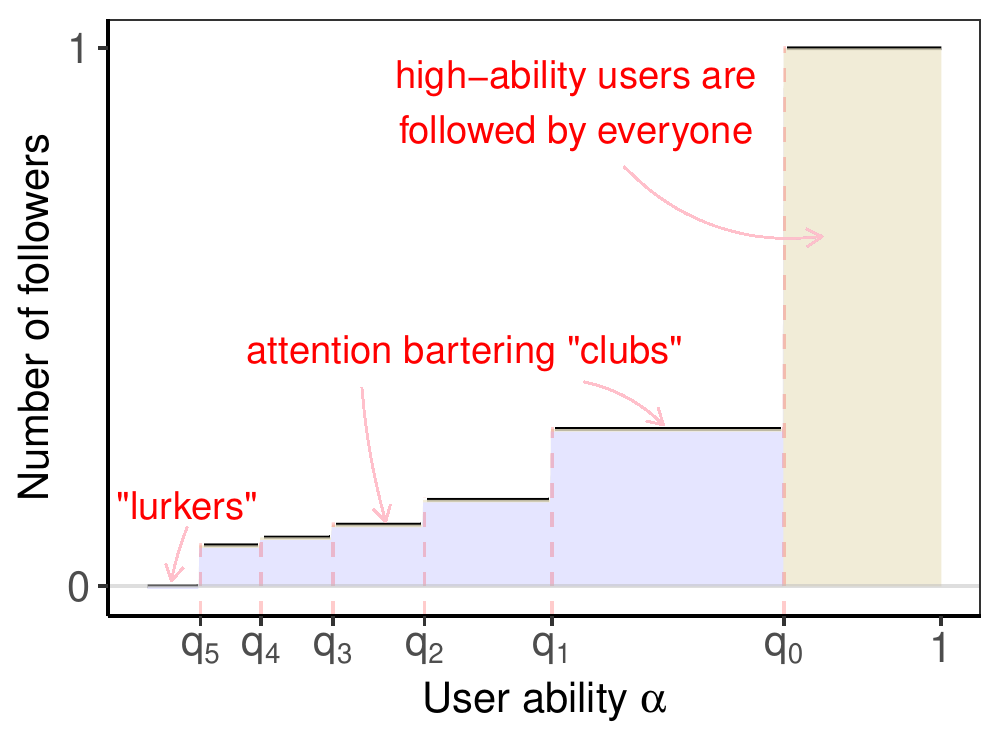}
		\label{fig:info_uniform_followers}
	\end{subfigure}
	\begin{subfigure}[t]{0.5 \textwidth}
  		\includegraphics[width = \linewidth]{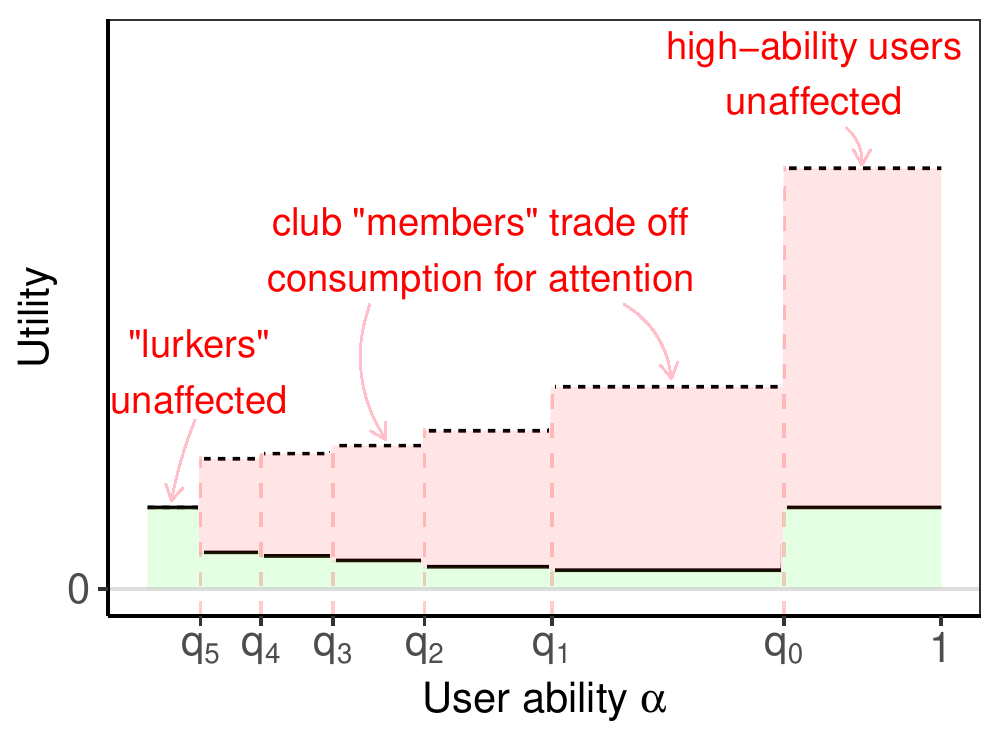}
  		\label{fig:info_uniform_utility} 
	\end{subfigure}
\end{subfigure}
\end{center} 
{\footnotesize \begin{minipage}{1 \linewidth}
\vspace*{-0.0in}
\emph{Notes:}
This figure plots users' followers, consumption utilities, and attention utilities, with and without attention bartering.
It illustrates the case of uniformly distributed valuations on the unit interval, $q_0 = 0.8$, and $I(x)=\sqrt{x}/2$.
In all panels, 
the yellow-shaded area depicts the number of organic followers, 
the blue-shaded area depicts the number of reciprocal followers,
the green-shaded area depicts consumption utility minus monitoring costs, 
and the red-shaded area depicts attention utility.
Panel~\ref{fig:no_info_uniform} plots the users' followers (left panel) and utilities (right panel) in a platform where attention bartering is impossible.
Only high-ability users are followed, and those users are followed by every user.
Every user obtains the same consumption utility, but only high-ability users obtain attention utilities.
Panel~\ref{fig:info_uniform} plots the users' followers (left panel) and utilities (right panel) in a platform with attention bartering.
The monitoring cost parameter is set to  $c=0.2$.
Attention bartering clubs emerge, in which users trade off decreases in consumption utility for increases in attention utility.
The number of lurkers decreases, and both lurkers and high-ability users experience no utility change. 
Figure~\ref{fig:econtwitter_vs_theory_statistics} reports additional quantities for the attention bartering equilibrium.   
\end{minipage} }
\end{figure}

\subsection{Attention bartering is possible}
\label{sec:public}
We next examine the equilibrium outcomes when attention bartering is possible---this is case when follower information is publicly accessible, thus enabling reciprocal following.
In terms of our model parameters, the monitoring cost $c$ is small enough that some users are incentivized to engage in attention bartering.
 
To derive the equilibrium in this case, we need to make a number of observations.
A first observation is that organic following is invariant to whether attention bartering takes place---no user is better off not following high-ability users.
Formally, for any equilibrium profile $f^*$, $f^*_\alpha(x)=1$ for all $x>q_0$.
As such, users with ability $\alpha>q_0$ are followed by every user, do not engage in reciprocal following, and obtain the same consumption and attention utilities as in the case where attention bartering is impossible.
Only users with ability $\alpha\le q_0$ have incentive to engage in reciprocal following. 
Proposition~\ref{prop:info_equilibrium_continuum} derives the structure of optimal reciprocation.
\begin{proposition}  \label{prop:info_equilibrium_continuum}
	In any equilibrium $f^*$, users with ability $\alpha \leq q_0$ will reciprocally follow users of abilities in a continuum $(\underline{q}(\alpha), \overline{q}(\alpha)]$, that is, $f^*_\alpha(x) = 1$ if $x \in (\underline{q}(\alpha), \overline{q}(\alpha)]$, and $f^*_\alpha(x) = 0$ otherwise.
	Furthermore, for any users $\alpha_1$ and $\alpha_2$ such that $\alpha_1>\alpha_2$, $\underline{q}(\alpha_1) \geq \underline{q}(\alpha_2)$ and $\overline{q}(\alpha_1) \geq \overline{q}(\alpha_2)$. 
\end{proposition}

\begin{proof}  
Consider any equilibrium where user $\alpha_1$ engages in reciprocal following, but does not employ a continuum strategy.
This implies that user $\alpha_1$ does not reciprocally follow some user $\alpha_2$, but reciprocally follows some user $\alpha_3 < \alpha_2$. 
User $\alpha_1$ would then be better off reciprocally following user $\alpha_2$, and incur a smaller consumption decrease.
For the same reason, all users with ability higher than $\alpha_1$ that reciprocally follow user $\alpha_2$ will also reciprocally follow user $\alpha_1$; this proves the second part of the proposition.
\end{proof}
\noindent Under profile $f$, a user $\alpha$ who reciprocally follows increases her utility by
\begin{equation}
	R_\alpha(f) = R \left (\underline{q}(\alpha),\overline{q}(\alpha) \right ) 
	= \underbrace{\int_{\underline{q}(\alpha)}^{\overline{q}(\alpha)} x - q_0 dx}_{\mbox{{\footnotesize consumption util. decrease}}} - 
	\underbrace{\int_{\underline{q}(\alpha)}^{\overline{q}(\alpha)} c dx}_{\mbox{{\footnotesize monitoring costs}}} + 
	\underbrace{ \vphantom{\Bigg [}I \left ( \overline{q}(\alpha) - \underline{q}(\alpha) \right )}_{\mbox{{\footnotesize attention util. increase}}}\!\!\!\! ,
\end{equation} 
for a total utility $ V_\alpha(f) = U_\alpha(f) + R_\alpha(f)$.
Proposition~\ref{prop:info_equilibrium_structure}  constructs the unique equilibrium.
\begin{proposition} \label{prop:info_equilibrium_structure}
	In a platform with attention bartering, there exists a unique equilibrium $f^*$, that can be computed as follows
	\begin{enumerate}
		\item For any $\alpha$, $f^*_\alpha(x) = 1$ for all $x>q_0$. 
		\item For users with ability $\alpha > q_0$, $f_\alpha^*(x)=0$ for all $x \leq 0$.
		\item Set $\alpha_1 = q_0$. 
		For all users with ability $\alpha_1$, $f_\alpha^*(x) = 1$ iff $x \in (\underline{q}(\alpha_1), \overline{q}(\alpha_1)]$, where $\overline{q}(\alpha_1)=\alpha_1$ and $\underline{q}(\alpha_1) = \argmax_{ x \in [0, \overline{q}(\alpha_1)] } \left \{ R \left (x,\overline{q}(\alpha_1) \right ) \right \}$.		Users with ability $\alpha \in (\underline{q}(\alpha_1), \overline{q}(\alpha_1)]$ reciprocally follow each other, and only each other.
		\item  Iterate by setting $\alpha_{i+1} = \underline{q}(\alpha_i)$, and repeating the previous step.
		\item The process stops when either $\underline{q}(\alpha_i) = \overline{q}(\alpha_i)$, or $\underline{q}(\alpha_i) = 0 $.
	\end{enumerate}
\end{proposition}

\begin{proof}
	Amongst users who reciprocate, users with ability $\alpha=q_0$ are the most ``valuable'' ones, as they decrease their reciprocal followers' consumption utility by the least amount.
	In any equilibrium profile, users that are reciprocally followed by those users will reciprocate---reciprocally following a higher ability user always is utility-improving for a lower-ability user, so long as it is utility-improving for a higher-ability user.
	As such, any profile where these users do not follow the highest-ability reciprocal followers can be improved upon, and is not an equilibrium profile.
	Because users are homogeneous, if a user belongs to the continuum of a higher-ability user, then that user's continuum is identical.
	Any continuum that stops before marginal utility starts decreasing can be improved upon, and is not an equilibrium.
	Evoking the same arguments, uniqueness is obtained by contradiction.
\end{proof}

In Figure~\ref{fig:no_info_uniform}, we illustrate the platform equilibrium when attention bartering is possible, for the case of $I(x) = \sqrt{x}/2,$ $q_0 = 0.8$, and $c=0.2$.
Users with ability $\alpha>q_0$, are unaffected by  the possibility of attention bartering---they are followed by every other user, and they do not engage in reciprocal following.
The set of lurkers shrinks, as lower-ability users can improve their total utilities by forming attention bartering clubs.
Five attention bartering clubs form in our example, comprising users of abilities in the continua $[q_1,q_0]$, $[q_2,q_1]$,  $[q_3,q_2]$,  $[q_4,q_3]$, and  $[q_5,q_4]$.
Note that the extent of reciprocal following---the size of these clubs---gradually becomes smaller, as reciprocally following lower-ability users is costlier. 
Reciprocation stops when users no longer find it profitable to follow their own type: users with ability $\alpha \leq q_5$ maintain their lurker status, and continue to only consume.\footnote{This can be thought of as the (Groucho) Marx condition: the marginal user who does not engage in attention bartering does not want to be a member of any club that would have her as a member.}
Club members are better off in terms of total utility but their consumption utility decreases, as they are forced to follow users below the opportunity cost $q_0$.
Higher-ability users belong to more populous clubs, and hence incur higher monitoring costs and more pronounced consumption utility decreases.\footnote{
 Interestingly, these clubs have increasing inbreeding homophilly without any notion of costly search \citep{davis1970clustering, currarini2009economic}.
}

The model generates testable predictions, which we summarize in Proposition~\ref{prop:attention_bartering_properties_homo}:
\begin{proposition} 
\label{prop:attention_bartering_properties_homo}
In the equilibrium of a platform where attention bartering takes place,
\begin{enumerate} [label={(\roman*)}]
\item The number of followers and follower-to-followee ratio are weakly increasing in user ability. Only high-ability users have follower-to-followee ratios higher than one. \label{pprop:eqbm_followers_ratio}
\item The number of followees is concave in user ability: it weakly increases initially, and then decreases for high-ability users. High-ability users and lurkers , who follow an equal number of users as lurkers.  \label{pprop:eqbm_following}
\item Lurkers and high-ability users do not engage in attention bartering. Every other user's number of reciprocal followers is positive and weakly increasing in her ability. \label{pprop:eqbm_reciprocation}
\item Users comprising an attention bartering club have the same number of followers, followees, and ratios. The size of the clubs---and hence the magnitude of these statistics---increases weakly in user ability. \label{pprop:eqbm_clubs}
\end{enumerate}
\end{proposition}
\noindent We depict these predictions in Figure~\ref{fig:individual_outcomes_unif}, as well as in Figure~\ref{fig:econtwitter_vs_theory_statistics} and Figure~\ref{fig:econtwitter_vs_theory_reciprocation} where we compare the model predictions with estimates from the \ET{} data.

\subsection{Preference heterogeneity} \label{ssec:heterogeneity}
Our basic model considers users that differ only in their abilities.
This can be a reasonable approximation for social media communities, but real-life users are surely heterogeneous in more ways than one.
It is worth examining how introducing additional sources of heterogeneity affects the predictions of the model.

Users may have heterogeneous consumption preferences.
One way to model this type heterogeneity is to introduce a horizontal consumption component to the model, in addition to the vertical component---user abilities.
Appendix~\ref{app:horizontal_model} develops one such extension, where users face idiosyncratic and non-symmetric  opportunity costs to consumption.
If the horizontal component outweighs substantially the vertical component, users of very different abilities may now follow each another, either organically or reciprocally, introducing inter-club connections. 
Importantly, because higher-ability users attract more organic followers, they engage in less reciprocation, and hence they incur less pronounced consumption utility decreases.
However, if consumption preferences correlate strongly with user abilities, attention bartering still takes place predominantly between users of similar abilities.

Users may also differ in how much they value attention.
All other factors equal, a user $\alpha$ that obtains higher attention utility from attracting followers will attempt to engage in reciprocation with users with ability lower than her club's lower bound $\underline{q}(\alpha)$.\footnote{In the extreme, a user whose preference for attention far outweighs her preference for consumption will reciprocally follow any willing user.
These are likely the social media users who openly advertise their ``follow-for-follow'' policies.}
These users will follow her back, as her ability is higher than the abilities of other members of their own club.
Having gained a more valuable reciprocal follower, these users may in turn unfollow their lowest-ability reciprocal followers.
This effect ripples through the population of reciprocal followers, ``smoothing out the edges'' of attention bartering clubs, and introducing inter-club connections.
Introducing heterogeneous monitoring costs has identical effects.

\OCD
\section{Evidence from \ET} \label{sec:empirics}
We next assess empirically core assumptions and predictions of our model.
Toward that end, we collected from \ET{}, a Twitter community comprising  professional and amateur users who tweet mostly about economics, and often follow each other.\footnote{On Twitter, users may prepend the hashtag character to relevant keywords or phrases within their tweets, as a form of tagging that enables cross-referencing of content sharing a theme or subject. 
Twitter communities often derive their names from recurring hashtags that their members use to broadcast messages to the community. 
For more details, see ~\url{https://help.twitter.com/en/using-twitter/how-to-use-hashtags}, and~\url{https://en.wikipedia.org/wiki/Hashtag}.}
Our rationale for focusing on \ET{} is to obtain data that closely matches the setting of our model, that is, data on users with higher preference homogeneity than what would be expected in a random sample of Twitter users.

It is important to emphasize the nature of this empirical exercise.
We are not fitting the \ET{} data to our model.
Instead, what we are doing is comparing the patterns in the \ET{} data---specifically network statistics as a function of user ability---to the patterns predicted in our model.
For example, our model predicts a ``followee dip'' at the highest levels of user ability---we look for this dip in the \ET{} data.
This exercise is somewhat similar to comparing the coefficients in a regression to the comparative statics of a model, except that we are making this comparison at many ability levels, where the signs of some predictions change i.e., followees is increasing in ability up until the dip, at which point the sign changes. 

\subsection{Data collection and sample definition \label{ssec:sample}}
Because there is no ``tag'' on Twitter indicating that a user is an economist, we adopted a heuristic process based on Twitter's list feature.
Lists are groups of Twitter users created and curated by individual users, in order to organize tweets thematically.
Lists can be made private or public, and can be shared with other users.
Importantly, adding a user to a list does not affect follower or followee counts, and hence prospective followers cannot see the number of lists to which a user has been added.\footnote{For more details on Twitter lists, see~\url{https://help.twitter.com/en/using-twitter/twitter-lists}.}

We began by collecting data on users comprising \RP's ``Economists on Twitter'' list.
We then collected data on users followed by at least $3$ users on the \RP{} list, for a total of~\numUsersNotclean{} users.
The data features for each user include the unique identifiers of her followers and followees, the number of tweets she has produced, the number of tweets she has liked, her tenure on the platform, and the number of lists that she has been added to by other users.
All data was collected using Twitter's public API.\footnote{The \RP{} list can be found at~\url{https://ideas.repec.org/i/etwitter.html}.
The official documentation describing the features of  data obtained through Twitter's API can be found at~\url{https://developer.twitter.com/en/docs/twitter-api/v1/data-dictionary/object-model/user}.
}

We refined the data sample as follows.
First, we removed ``inactive'' users, defined as users with fewer than \followerLower{}  followers, fewer than \followingLower{} followees, or fewer than \tweetLower{} tweets. 
Second, we removed ``superstar'' users, defined as users with more than \followerUpper{}  followers.
This restriction removes few economist outliers, but many non-economists who are particularly popular with \ET{} users.\footnote{For example, Barack Obama, Shakira, and Britney Spears are the three most followed non-economists who are followed by 3 or more users in the \RP{} list. 
Among economists in the \RP{} list, the ``superstar'' restriction resulted in Paul Krugrman, Xavier Sala-i-Martin, Joseph Stiglitz, and Alejandro Gaviria being dropped from the sample. 
Erik Brynjolfsson (barely) made it.}
Third, to eliminate some ``spam'' accounts from our data, we removed users with more than \followingUpper{} followees.
This restriction only drops a small number of users, and is consistent with the Twitter-imposed basic followee limit.\footnote{For more details, see~\url{https://help.twitter.com/en/using-twitter/twitter-follow-limit}.}
The final sample consists of \numUsers{} \ET{} users, who collectively follow and are followed by \ALLnumUsers{} users.

\subsection{Followers, followees, and ratios}
Figure~\ref{fig:econtwitter_stats_distr} reports the distributions of followers, followees, and follower-to-followee ratios of \ET{} users.
We plot a Gaussian kernel density estimate of the distribution of each outcome, with the bandwidth selected using  Silverman's rule of thumb~\citep{silverman1986density}.
The red vertical lines depict the median (solid) and the mean (dashed) of each distribution.
\begin{figure}[h!]
\begin{center}
\begin{minipage}{1 \linewidth}	
\caption{Distributions of followers, followees, and ratios of \ET{} users}
\label{fig:econtwitter_stats_distr}
\includegraphics[width = 1 \textwidth]{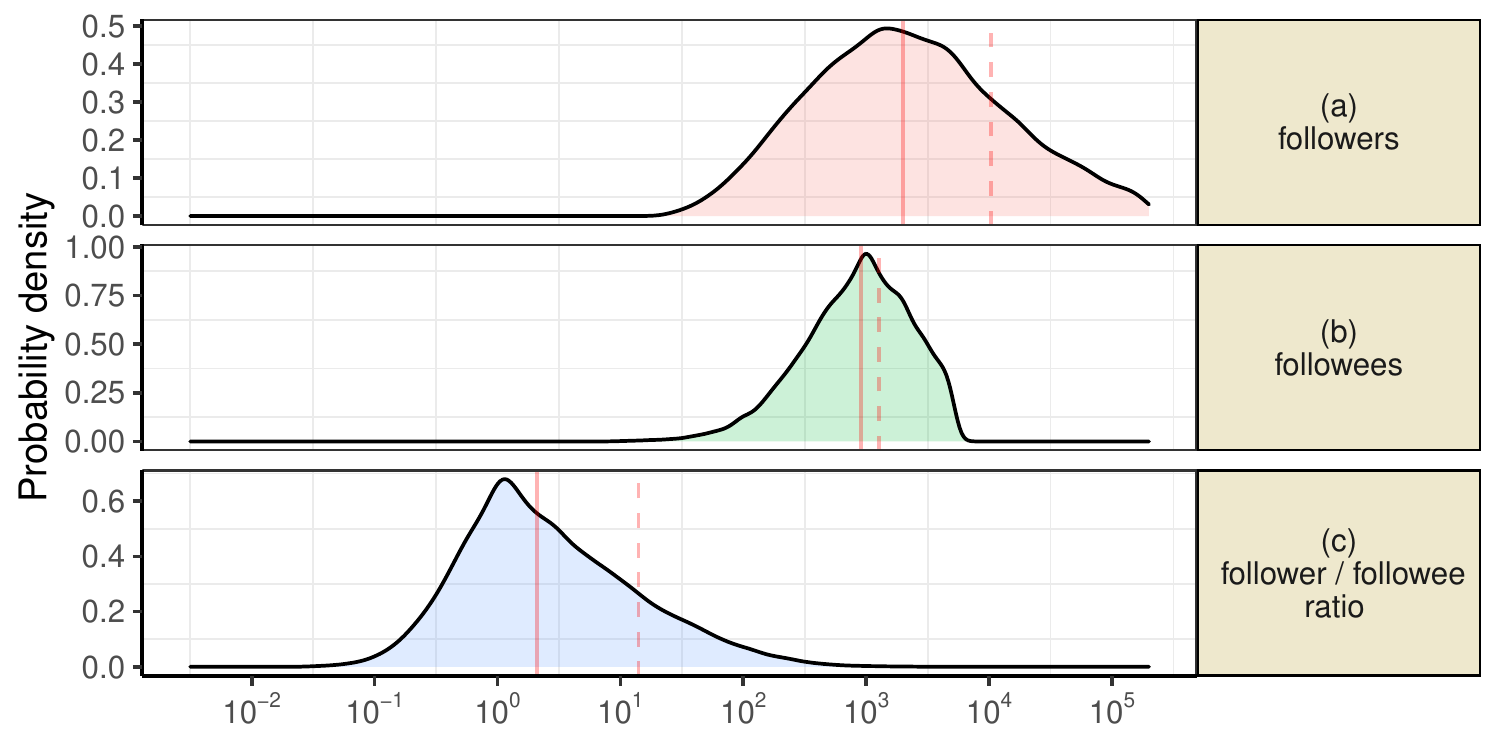}
\end{minipage}
\end{center}
{\footnotesize \begin{minipage}{1 \linewidth}
\vspace*{-0.05in}
\emph{Notes:}
This figure reports kernel density estimates of the probability density function of the distributions of followers, followees, and follower-to-followee ratios for \ET{} users.
For each facet, the vertical red lines depict the median (solid) and the mean (dashed) of the corresponding distribution, and the kernel bandwidth is selected using  Silverman's rule of thumb~\citep{silverman1986density}.
See Appendix~\ref{app:empirics} for density estimates of the same statistics for all users.
\end{minipage}}
\end{figure}

On average, an \ET{} user has \meanFollowers{} followers, \meanFollowing{} followees, and a ratio of \meanRatio{}, whereas the median user has \medianFollowers{} followers, \medianFollowing{} followees, and a ratio of \medianRatio{}.\footnote{In Appendix~\ref{app:empirics}, we report estimates computed on the entire sample. 
Compared to the average Twitter user, \ET{} users have on average more followers and followees, and higher ratios: the median user in the entire data set has \ALLmedianFollowers{} followers, \ALLmedianFollowing{} followees, and a ratio equal to \ALLmedianRatio.
}

Follower counts are highly right-skewed, with some ``stars'' that attract vast numbers of followers.
This is consistent with producers facing no distribution costs and being vertically differentiated, with the ``best'' users having large audiences.
In contrast, the distribution of users' followees is substantially less right-skewed, and does not have a long tail.
This is consistent with consumers having scarce attention.

Despite this follower/followee distinction---and some star users having enormous ratios---follower-to-followee ratios are clustered around one.
This is consistent with Proposition~\ref{prop:attention_bartering_properties_homo}~\ref{pprop:eqbm_followers_ratio}.
This clustering of ratios near one has been documented in previous work spanning several social media, and a relatively long period of time~\citep{kwak2010twitter, myers2014information, sadri2018analysis}.
Interestingly, our ``clubs''---where all users follow each other---have ratios slightly less than one, with the denominator inflated by the ``stars'' those users follow.
If there are few stars, our model predicts large numbers of users with ratios of approximately one.

\subsection{Constructing a proxy for user ability using lists} \label{ssec:user_ability}
The key distinguishing feature of our model is the predictions it makes about ability and outcomes.
However, the ``abilities'' of Twitter users cannot be observed directly. 
As proxy for user ability, we use the number of lists to which a user has been added by other users.
Lists have several properties that are useful for constructing a proxy for user ability.
First, adding a user to a list indicates interest in her tweets, without increasing the enlister's followee count.
Second, the number of lists to which a user has been added is not observable by her would-be followers, and hence list-adding cannot be bartered similarly to attention.

One limitation of the list count is that users with greater tenure on the platform are likely to have been added to more lists.  
To remove some of the tenure dependency effect, we residualize the log number of lists to which a user has been added, by partialing out the log number of tweets that the same user has generated.
We hereafter refer to  user ability as the percentile a user belongs to in our data with respect to the (residualized) number of lists to which she been added.

\subsection{User ability and network statistics}
Our model makes sharp predictions about users' network statistics in the attention bartering equilibrium.
In Figure~\ref{fig:econtwitter_vs_theory_statistics}, we compare these predictions to estimates obtained from the \ET{} data.
Each row depicts a separate network statistic.
The left-column facets depict the empirical estimates as a function of the list-defined user ability.
Each point represents the mean value for users belonging to the same ability percentile, and a 95\% confidence interval is plotted around each point estimate.
Outcomes are normalized to the unit range---see Appendix~\ref{app:empirics} for an unscaled version of the empirical estimates.
The right-column facets depict the equilibrium predictions as a function of user ability, using the model parameterization of Figure~\ref{fig:individual_outcomes_unif}.
The model parameters were chosen for ease of exposition---what matters for our empirical exercise is the shape of the equilibrium predictions.

Facet (a) plots the number of users' followers.
Followers increase in users' abilities, suggesting that our choice of proxy is reasonable.
Consistent with our modeling assumption and with  Proposition~\ref{prop:attention_bartering_properties_homo}~\ref{pprop:eqbm_followers_ratio}, high-ability users attract a far larger number of followers than low-ability users.

\begin{figure}[h!]
\begin{center}
\begin{minipage}{0.972 \linewidth}	
\caption{\ET{} estimates and model predictions of network statistics. \label{fig:econtwitter_vs_theory_statistics}}
\includegraphics[width = 1 \textwidth]{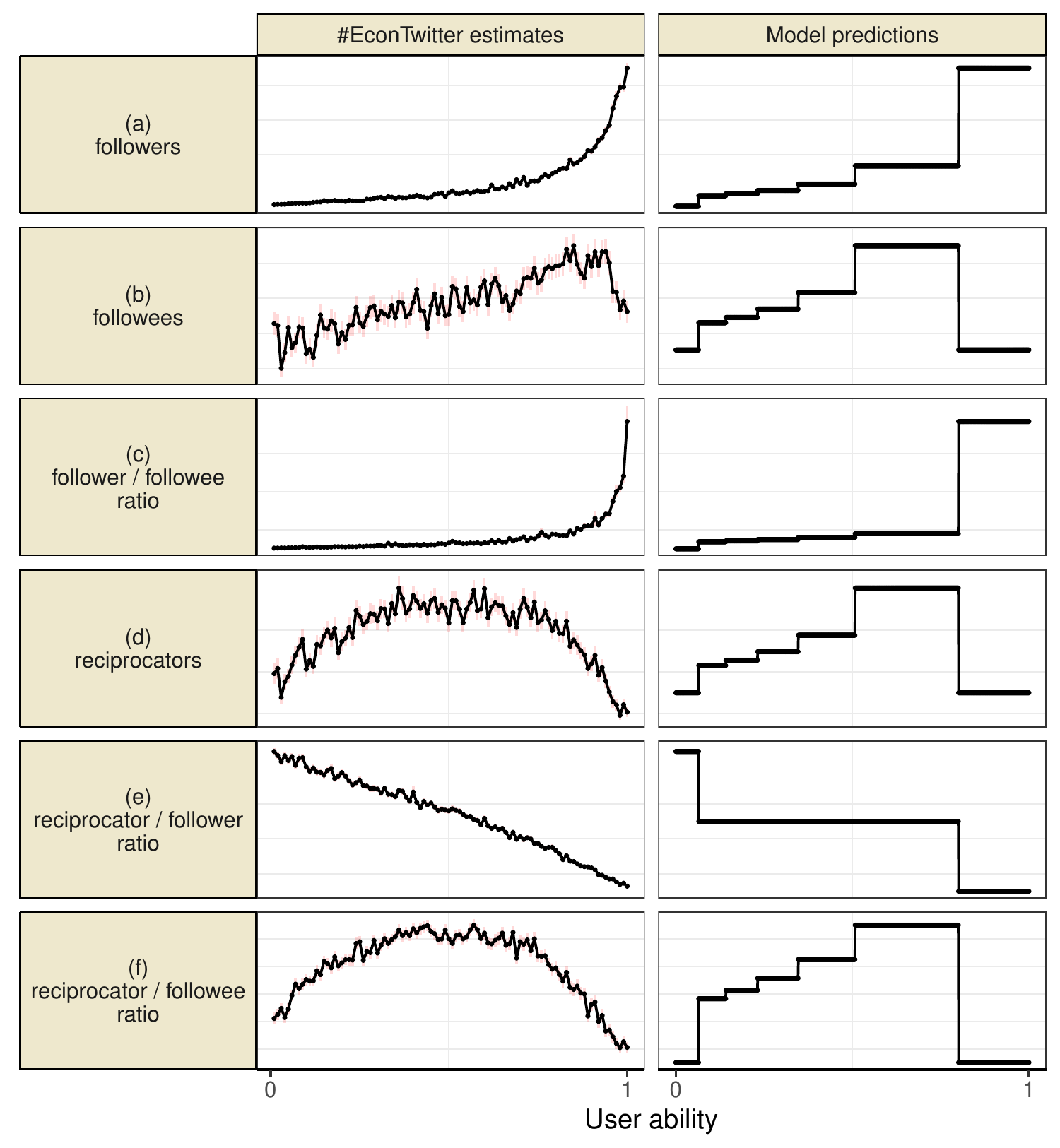}	
\end{minipage}
\end{center}
{\footnotesize \begin{minipage}{1 \linewidth}
\vspace*{-0.1in}
\emph{Notes:}
This figure reports empirical estimates and theoretical predictions for  networks statistics (y-axis) as a function of user ability (x-axis).
On the left-column, we plot estimates of \ET{} users' (a) followers, (b) followees, and (c) follower-to-followee ratios, (d) reciprocators, (e) reciprocator-to-follower ratios, and (f) reciprocator-to-followee ratios. 
Each point represents the mean value (y-axis) for users belonging to the same ability percentile (x-axis), and a 95\% confidence interval is shown around each point estimate.
On the right column, we plot model predictions for the same quantities.
The illustrated attention bartering equilibrium assumes uniformly distributed valuations on the unit interval, $q_0 = 0.8$, and $I(x)=\sqrt{x}/2$.
All quantities are normalized to the same scale.
See Section~\ref{ssec:user_ability} for details on the definition of abilities of \ET{} users, and Appendix~\ref{app:empirics} for the scales of the empirical estimates.
See Figure~\ref{fig:info_uniform} for an alternative representation of the same attention bartering equilibrium.
\end{minipage}}
\end{figure}

Facet (b) plots the number of followees---the number of users followed by each user.
Followees increase initially in user ability, but at a much slower rate than the number of followers, and then decrease at high ability levels.
This ``followee dip'' is consistent with Proposition~\ref{prop:attention_bartering_properties_homo}~\ref{pprop:eqbm_following}, and supports our modeling approach that following users is costly in terms of consumption, and not just a matter of awareness or opportunity.

Follower-to-followee ratios are reported in facet (c). 
These ratios are increasing in user ability, but at a slower rate than the number of followers: users within broad ability ranges have similar ratios, consistent with  Proposition~\ref{prop:attention_bartering_properties_homo}~\ref{pprop:eqbm_followers_ratio} and~\ref{pprop:eqbm_clubs}.\footnote{At high ability levels, ratios seemingly become more dispersed within the same ability level.
One possible explanation is that high-ability users, as captured by our proxy, comprise both star users with many followers and few followees, but also users with large numbers of both followers and followees.
The latter are likely users who have attained visibility due to attention bartering, and have been subsequently added to many lists.
This does illustrate a limitation of our ability proxy.}

\subsection{Reciprocation and attention bartering}
We cannot observe whether a relationship is truly reciprocal in the sense of our model---we do not know what would happen if one side of the deal unfollowed.
As such, we cannot verify whether a bilateral relationship was the result of the two users following each other organically, or due to attention bartering.
Nevertheless, we refer to any bidirectional following relationship in our data as reciprocal in what follows.

Facet (d) in Figure~\ref{fig:econtwitter_vs_theory_statistics} plots the number of users' reciprocal followers as a function of user ability.
The number of reciprocal followers initially increases in user ability, but then decreases, and eventually attains its minimum at high user abilities.
In facet (e), we see clearly that the reciprocator-to-follower ratio is strictly decreasing in user ability.
Low-ability users reciprocally follow the majority of their (few) followers, and users become more selective as their abilities increase.
At the extreme, users of very high abilities follow back fewer than 0.5\% of their followers.

Facet (f) reports the fraction of a user's followees that follow her back.
About 20\% of the followees of high-ability users follow them back, whereas the same number is about 40\% for medium-ability users.
The predictions of our model are borne-out in our \ET{} data: middle-ability users are the ones who barter the most, with their attention bartering intensity increasing in their abilities, while low- and high-ability users barter the least---consistent with Proposition~\ref{prop:attention_bartering_properties_homo}~\ref{pprop:eqbm_reciprocation}.

\subsection{Assortative matching}
We find strong evidence of assortative matching. 
Figure~\ref{fig:econtwitter_vs_theory_reciprocation} plots the ability ranges of users' organic and reciprocal followees as a function of user ability.
Panel~\ref{fig:econtwitter_stats_reciprocation_range} reports estimates from the \ET{} data. 
The left facet plots the distributions for users' organic followees---followees who do not follow back---and the right facet plots the distribution for users' reciprocal followees.
In each facet, we report the median (black dots) and the interdecile range (pink-shaded area) of the ability distribution of followees, for each ability level.
Panel~\ref{fig:model_stats_reciprocation_range} reports model predictions for the same quantities, reporting the median ability of users followees in the attention bartering equilibrium (black dots), and the ability range that users are willing to follow organically and reciprocally (pink-shaded area).

\begin{figure}[h!]
\begin{center}
\begin{minipage}{1 \linewidth}	
\caption{Ability ranges of users' followees\label{fig:econtwitter_vs_theory_reciprocation}}
\begin{subfigure}[t]{1 \textwidth}
\caption{\ET{} estimates}
	\includegraphics[width = 1 \textwidth]{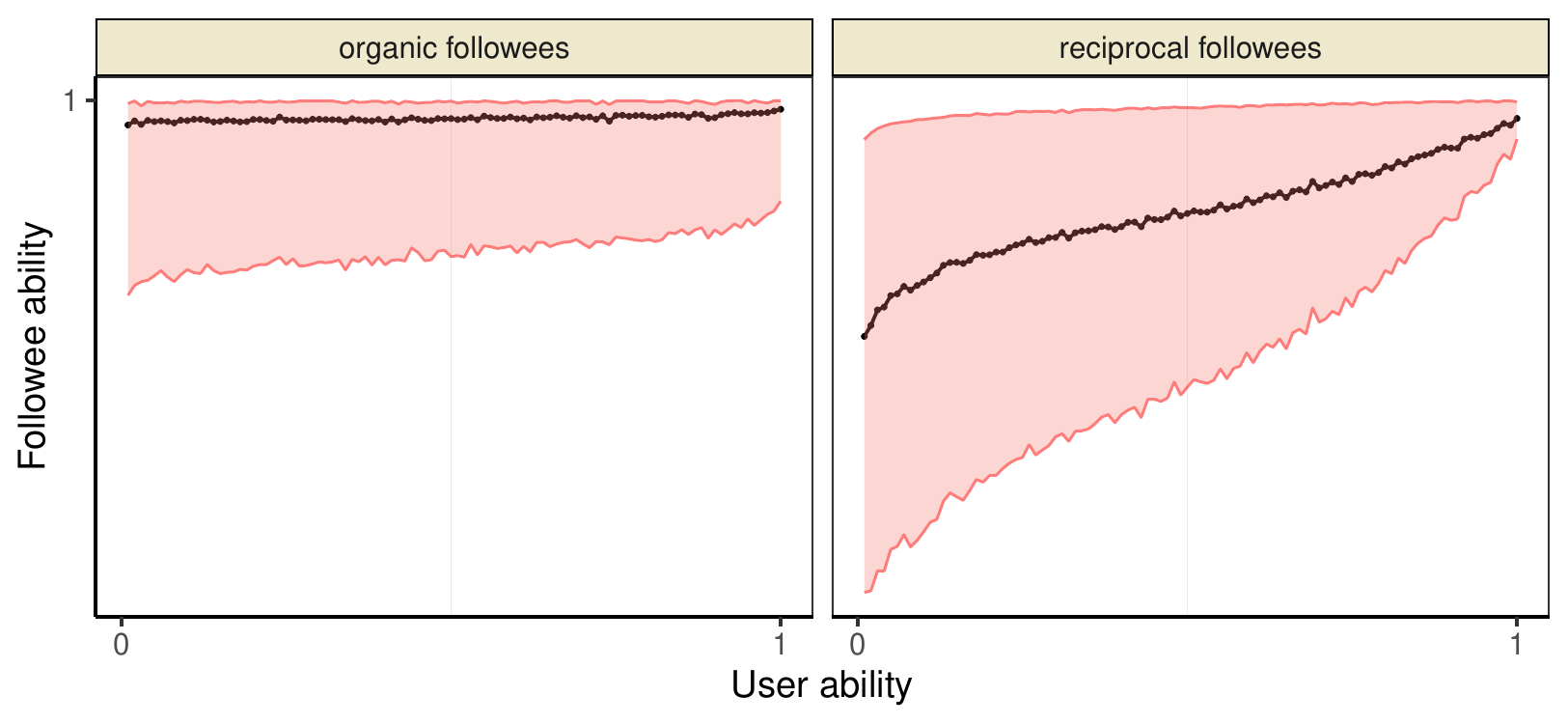}	
	\label{fig:econtwitter_stats_reciprocation_range}
\end{subfigure}
\begin{subfigure}[t]{1 \textwidth}
\centering
\caption{Model predictions}
	\includegraphics[width = 1 \textwidth]{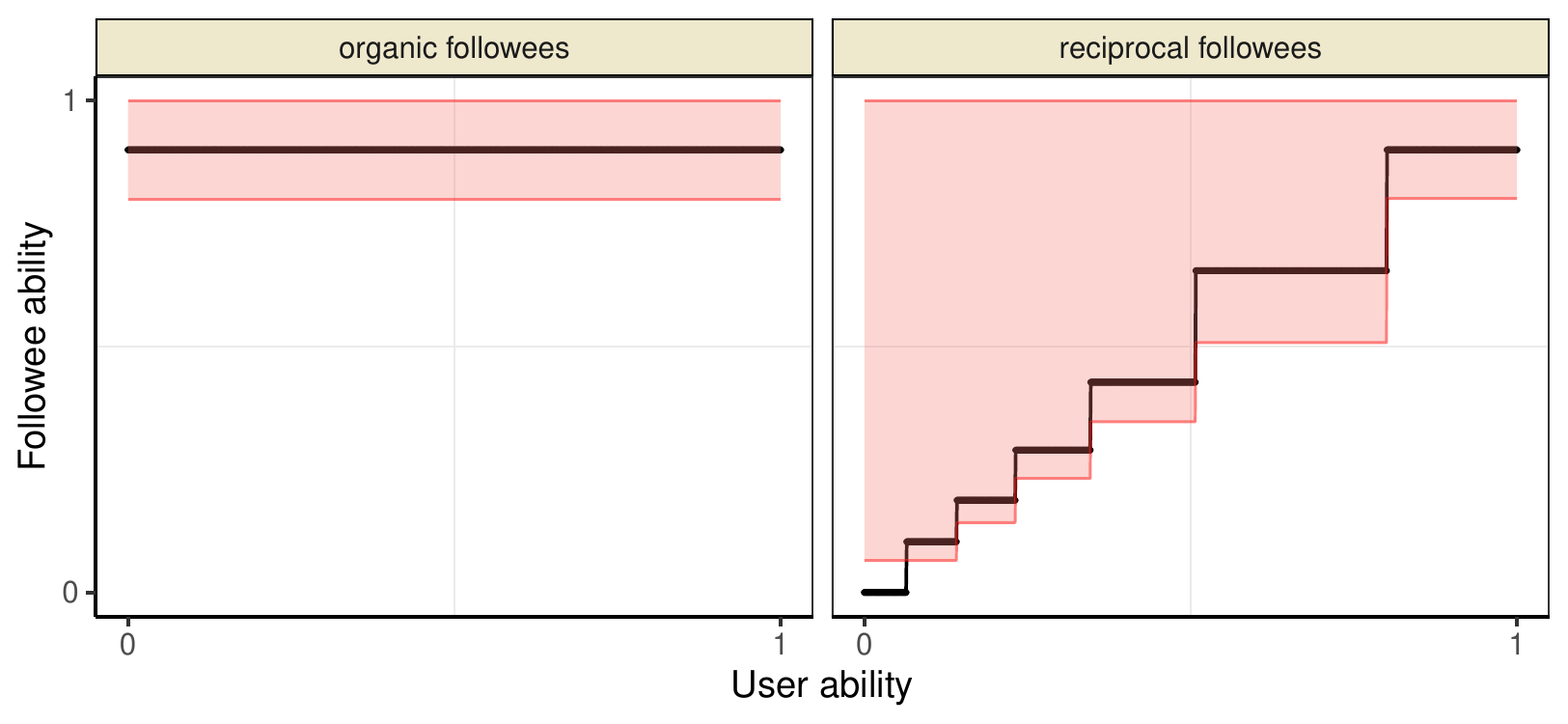}
	\label{fig:model_stats_reciprocation_range}
\end{subfigure}
\end{minipage}
\end{center}
{\footnotesize \begin{minipage}{1 \linewidth}
\vspace*{-0.2in}
\emph{Notes:}
This figure reports empirical estimates and theoretical predictions of the following behavior of users.
Panel~\ref{fig:econtwitter_stats_reciprocation_range} plots the median and interdecile range of the distribution of followee abilities (y-axis) for \ET{} users belonging to the same ability percentile (x-axis).
The left facet reports the distribution for organic followees---followees who do not follow the user, and the right facet for reciprocal followees---followees who are also followers.
Panel~\ref{fig:model_stats_reciprocation_range} plots model predictions for the same quantities.
The illustrated attention bartering equilibrium assumes uniformly distributed valuations on the unit interval, $q_0 = 0.8$, and $I(x)=\sqrt{x}/2$.
See Section~\ref{ssec:user_ability} for more details on the definition of abilities of \ET{} users.
See Figure~\ref{fig:info_uniform} for an alternative representation of the same attention bartering equilibrium.
\end{minipage}}
\end{figure}

The data estimates and the predictions of the model are remarkably similar. 
Low-ability users are willing to reciprocate with any willing user who is not of extremely low quality.
High-ability users are only form reciprocal relationships with other high-ability users---note that these relationships are bidirectionally organic.
The picture is completely different for organic followees: users follow unconditionally only high ability users, independent of their own abilities, both in our data and in our model.\footnote{Network formation models based on preferential attachment could explain some but not all of these patterns~\citep{barabasi1999emergence}.
For example, highly right-skewed follower distributions are consistent with rich-get-richer link formations, but would require more able users to have systematically joined the platform earlier.
Had this been the case---and assuming that users do not ``re-wire'' frequently---it would also explain why higher ability users tend to reciprocate with high ability users, as those are the users that were already present when they joined.
We find very little relationship between user tenure and our ability measure, and our ability measure is has substantially more predictive power than tenure with respect to users' follower-to-followee ratios.
For more details, see Appendix~\ref{app:empirics}.}

\OCD
\section{Discussion} \label{sec:discussion}

\subsection{Production, consumption, and quality of content} \label{ssec:quality}
Our model has several touch points with the economics of club goods~\citep{buchanan1965economic}.
Each club member pays for membership with their attention, in exchange for consuming a finite resource---the attention of other members.
In every club, members are not quite good enough for the club ``above,'' but are too good for the club ``below.''
The marginal club member of the lowest-ability club has ability equal to the marginal utility of receiving the attention of a single follower: this is exactly the first user whose ability is so low that she would not barter with herself.\footnote{On a somewhat more general context than social media, Groucho Marx famously made a similar observation: ``I don't want to belong to any club that will accept me as one of its members.''}
Particular to the social media context is the fact there is no mechanism excluding non-members from consuming the content produced within clubs---the scarce resource is not the content, but the attention  bartered between the club members.\footnote{Social media clubs have no officers, no explicitly stipulated rules, and (potentially) fractional memberships.}

In our model, attention bartering induces users of lower abilities to form clubs and to produce content---these users would have otherwise lurked. 
Consequently, the quantity of content produced on the platform increases, but its average quality decreases.
This low-quality content is consumed only by users belonging to the club wherein it was produced.
However, it can also spread to the platform through channels that include 
platform features such as algorithmic recommendation systems, 
new users who ``join'' these clubs unwittingly, 
and the inter-club connections created by preference heterogeneity (see Section~\ref{ssec:heterogeneity}).

If low-quality content is likely to be less truthful, attention bartering clubs become natural ``reservoirs'' of disinformation---these  can be thought of as the social media analogue to the epidemiological concept of reservoirs of pathogens~\citep{haydon2002identifying}.
New users may then pick up and spread this disinformation unintentionally.
Furthermore, to the extent that beliefs and behaviors are shaped by one's community, attention bartering may also exacerbate polarization. 

Attention bartering can have positive effects that are not strictly captured by our model.
Although in our base model attention bartering decreases the average quality of content produced, it may increase the average quality of content consumed in practice. 
For example, assume instead that every user has the potential to produce extremely valuable ``star'' content with a very low probability, which does not depend on her ability.
Attention bartering increases the number of users who actively produce, and hence the amount of star content produced on the platform.
Similarly, by introducing additional edges to the network, attention bartering may also increase the probability that star content will become ``viral.''\footnote{Presumably, attention bartering is to thank for---according to some---the best viral video of 2019, which can be accessed at~\url{https://twitter.com/indiemoms/status/1189587264654974980}.}

\subsection{Network statistics and user ability} \label{ssec:uncertainty}
Social media platforms typically display network information, such as the number of followers and followees, prominently within each user's home page (see Figure~\ref{fig:screenshots}).
This information presumably signals ability to prospective followers, and is conceptually similar to ``readership'' and ``viewership'' statistics reported by traditional media.
It is worth examining how attention bartering affects the informativeness of these signals. 

The two most common proxies for user ability on social media are a user's number of followers and ``ratio.''
The ratio of a user is defined as the number of users that follow her over the number of users she follows. 
In our base model, both statistics are weakly increasing in user ability;
only high-ability users have ratios greater than one, and users belonging to the same club have the same ratio and number of followers.
With heterogeneous preferences, the two statistics become strictly increasing in users' abilities, and hence likely more informative (see Appendix~\ref{app:horizontal_model}).
Attention bartering increases the ratio of users with more organic followees than followers, and vice versa for users with more organic followers than followees---as the number of reciprocal followers increases, ratios converge to one in both cases.
This finding helps to explain why social media users attempt to emulate the equilibrium outcome, keeping their ratios close to one.

The informativeness of these signals can break down if users have heterogeneous attention preferences, or heterogeneous monitoring costs.
For example, if the disutility of attention bartering is low enough, a low-ability user who values attention greatly can form a large number of reciprocal following relationships, potentially making her number of followers and ratio greater than those of a higher-ability user who values attention less.
Furthermore, real-life users may fail to monitor successfully that their reciprocal users have unfollowed them.
Consequently, a user who ``follows-then-unfollows'' unrelentingly may achieve statistics that resemble those of a high-ability user.\footnote{This is likely the reason why social media platform TikTok  displays prominently the amount of  ``likes'' each user's content has garnered within her home page.}

\subsection{Platform strategy and market design}
Social media platforms may be better off allowing---or even encouraging---attention bartering.
To see why, recall that a user who engages in attention bartering experiences higher utility from using the platform, and hence she is more likely to remain active on the platform.
By remaining active, she then increases the attention utilities of users she follows either organically or reciprocally, as well as the consumption utilities of her organic followers.
This reinforces the underlying network effect.
Today, the business model of social media platforms typically hinges upon monetizing user activity through advertising, and hence more or more active users likely result in higher profits---at least in the short run.

As social media platforms mature, they are likely to address some of the downsides of attention bartering.
Social media platforms regularly employ mechanisms aimed at increasing the consumption utility of users.
A side effect of these mechanisms is that they alleviate the consumption utility penalty of attention bartering, thereby making it more compelling. 
Examples include
``lists'' that allow users to consume content generated by a subset of the users---even ones they do not follow,  
the option to ``mute'' content generated by a user without unfollowing her, and 
algorithmically curated rankings that display content in non-chronological order. 
However, because these features affect directly the attention utility users receive, they may also affect adversely the production incentives of non-bartering users.

Attention bartering is sustained in equilibrium because users choose to incur monitoring costs to ensure that their reciprocal followers hold their end of their bargain.
Increasing these monitoring costs can decrease  the amount of attention bartering taking place.
Platforms could make it harder to barter attention by ceasing to notify users who attempt to initiate such relationships that their prospective reciprocal followers have followed them back.
Post the relationship formation period, platforms could also increase users' costs of finding out whether other users follow them.
For example, Twitter could remove the ``Follows you'' tag (see Figure~\ref{fig:screenshots}), and  make it costlier---or impossible---to sift through users' lists of followers and followees.
To dissuade ``follow-then-unfollow'' strategies, platforms could impose a time limit before being able to unfollow a user, and display prominently the number of unfollowed users within each user's home page.

\subsection{Social and individual welfare}

The potentially negative effect of social media on individual welfare has received substantial attention.
Millions of people clearly enjoy using social media---or at least prefer them to their next best alternatives~\citep{pew2019social}.
Social media can amplify individuals' reach, 
diversify news consumption,
facilitate participatory discourse, 
and influence public opinion positively,
but they may also increase political polarization, 
and serve as a vector for misinformation~\citep{george2003affects, tufekci2017twitter, Vosoughi1146, alatas2019,  allcott2020welfare, levy2020social}.

In an interesting experiment that speaks to individual welfare, \cite{allcott2020welfare} find that users who were paid not to use Facebook reported being happier, and used social media less even after the experiment ended.
They point out that the pattern they observe is consistent with social media usage being an addictive good in the \cite{becker1988theory} sense.
If it was the \emph{consumption} of social media \emph{per se} that made people unhappy, it would be hard to understand why anyone used social media in the first place.
But our ``two kinds of utilities'' perspective fits the fact nicely---social media consumption makes people unhappy (at at least not as happy if they consumed only what they truly wanted), but it is necessary to gain the pleasurable attention that consumption obtains.
It is interesting to note that Becker himself viewed people as the most addictive ``good.''\footnote{See~\url{https://freakonomics.com/2008/11/11/gary-becker-thinks-the-most-addictive-thing-is/}.}

\OCD
\section{Conclusion} \label{sec:conclusion}
This paper offers an economic model that can explain otherwise puzzling patterns of consumption and production on social media.
In our model, attention bartering profoundly affects the production and consumption of content on social media platforms: users of similar abilities form clubs and actively produce content, and consume the content of members of the same club---but they would have otherwise not produced any content.
Social media looks quite different from regular media, and the difference is the possibility of attention bargaining.

We show empirically that many of the features of a prominent Twitter community are explicable through an economic, attention bartering lens.
The empirical findings that we report are consistent with attention bartering taking place on \ET{}, and our model predictions are largely borne out.
In both our model and in our data, medium-ability users are the keenest attention barterers, forming large fractions of their relationships with their attention utilities in mind.
On the other hand, low- and high-ability users follow others seemingly to only increase their consumption utilities, albeit for very different reasons: low-ability users would like to barter but lack the opportunities to do so, and bartering opportunities abound for high-ability users---who have no incentive to barter.

Attention bartering emerges because of the coexistence of production and consumption incentives for users.
As such, it can occur even in the absence of algorithms steering users to form connections.
In contrast to approaches that point to algorithms as the culprit, we offer an incentives-based model of network formation, with algorithms playing no real role.
This is simplification of course, but much of the discussion about social media has focused on the role played by algorithms, with relatively little attention paid to incentives.
The degree to which the main patterns in the \ET{} data are explicable only with an incentives perspective suggests that a focus on algorithmic explanations is at the very least incomplete.

In addition to offering a description of social media, our paper illustrates a core market design challenge for would-be social media platform designers.
Platforms may try to harness attention bartering for good but reduce its negative effects.
We lay out several market design interventions that can potentially address the downsides of attention bartering;
future research could evaluate these mechanisms empirically, or propose other ways to dissuade social media users from playing the ensuing ``game of follows.''

\newpage \clearpage
\bibliographystyle{aer}
\bibliography{twitter.bib}

\appendix

\setcounter{section}{0}
\renewcommand\thesection{\Alph{section}}
\setcounter{equation}{0}
\renewcommand{\theequation}{A\arabic{equation}}

\clearpage \newpage
\section{Heterogeneous consumption preferences}
\label{app:horizontal_model}
We extend the model of Section~\ref{sec:model} to allow for horizontal consumption preferences.
Our modeling approach is to assume that a user's opportunity cost for following another user is independently and identically drawn from a given distribution.
As such, there is now both a vertical component (user ability) and a horizontal component (idiosyncratic opportunity costs) to consumption preferences.
Users no longer share one universal ordering of abilities; rather, each user has an individual preference ordering for the other platform users.
It is worth noting that high-ability users are more likely to rank higher in a user's ordering, as a vertical component still exists.
For ease of exposition, we fix the opportunity cost distribution to be the uniform distribution on the unit interval in what follows.

Because there is no universal ordering of user abilities, each user now decides on the maximum consumption utility decrease that she is willing to incur in order to engage in a reciprocal following relationship. 
We adjust the strategy notation, and denote by ${f}_\alpha \in [0,1]$ the maximum consumption utility decrease that user $\alpha$ is willing to incur to reciprocally follow another user. 
If $f_\alpha = 0$, user $\alpha$ follows only organically, whereas if $f_\alpha > 0$, user $\alpha$ is willing to engage in reciprocal relationships.

\subsection{Attention bartering is impossible}
When attention bartering is not possible, reciprocal following relationships cannot be sustained.
There exists a unique equilibrium $f^*$, where $f^*_\alpha = 0$ for all users.
In this equilibrium, only organic followers exist.
The probability that a user $\alpha$ will organically follow a user $x$ is $\Pr ( q_0 \leq x) = x.$
User $\alpha$ follows
\begin{equation}
n_\alpha(f^*) = \int_0^1 x\;\;dx = \frac{1}{2}
\end{equation} 
users, and obtains consumption utility
\begin{equation}
U_\alpha(f^*) = \int_{0}^1 E [ x - q_0 | q_0 \leq x ] \Pr(q_0 \leq x) dx
		= \int_0^1 \frac{x^2}{2} da = \frac{1}{6},
\end{equation}
where $E [ x - q_0 | q_0 \leq x ]$ is the expected consumption utility obtained conditional on a user with ability $x$ being followed organically.
Similarly, user $\alpha$ attracts
\begin{equation}
	m_\alpha(f^*) = \int_0^1 \alpha \;\; dx = \alpha 
\end{equation}
followers, from whom she obtains attention utility $I_\alpha(f^*) = I(\alpha)$.

\subsection{Attention bartering is possible}
When attention bartering is possible, user $\alpha$ may increase her threshold $f_\alpha$, thereby trading off a decrease in her consumption utility for an increase in her attention utility.

By symmetry, users of the same ability will choose the same  strategy in any equilibrium---recall that abilities are uniformly distributed on the unit interval.
Consider a user $\alpha$, and a  profile $f = (f_\alpha, f_{-\alpha})$.
Under this profile, user $\alpha$ has a number of reciprocal followers equal to
\begin{align}
\label{eq:het_rec_followers}
r_\alpha(f) = 	
& \int_0^1 (1-\alpha)(1-x)\Pr(x + f_a > q_0 | x < q_0) \Pr(\alpha + f_{x} > q_0| \alpha <q_0) dx \nonumber \\ = 
& \int_0^1
	\underbrace{(1-\alpha)(1-x) \min \left \{ \frac{f_\alpha}{1 - x}, 1 \right \} \min \left \{ \frac{f_{x}}{1 - \alpha}, 1 \right \} }_{\mbox{{\footnotesize $r_{a,x}(f)$}}} dx.
\end{align}
It is worth examining the terms of Equation~\eqref{eq:het_rec_followers}.
Consider some ability level $x$, and observe that user $\alpha$ is followed organically by an $\alpha$ fraction of users of any ability level, including users with ability $x$.
Amongst the remaining $1-\alpha$ users that do not follow her, user $\alpha$ follows a fraction $x$ organically herself, leaving a fraction $(1-\alpha)(1-x)$ to be candidates for reciprocal following.
Amongst these potential relationships, a fraction $\min \left \{ \frac{f_\alpha}{1 - x}, 1 \right \} \min \left \{ \frac{f_{x}}{1 - \alpha}, 1 \right \}$ materializes for follow strategies $f_a, f_x$, giving a total of $r_{a,x}(f)$.

Under the profile $f$, user $\alpha$ follows
$n_\alpha(f) = \frac{1}{2} + r_\alpha(f)$
users, and is followed by 
$
	m_\alpha(f) = \alpha + r_\alpha(f)
$
users, and obtains total utility
\begin{align}
	V_\alpha(f) & = \frac{1}{6} + \underbrace{\int_0^1 r_{\alpha,x} (f) E [x-q_0| q_0 - f_\alpha \leq x \leq q_0 ] dx}_{\mbox{{\footnotesize consumption utility decrease}}} - c \int_0^1r_{\alpha,x}(f)dx  + I \left ( \alpha + r_\alpha(f) \right ) \\ 
	& = \frac{1}{6} - \int_0^1 r_{\alpha,x} (f) \frac{f_\alpha}{2} dx - c \int_0^1r_{\alpha,x}(f)dx + I \left ( \alpha + r_\alpha(f) \right ) \\
	&=  \frac{1}{6} - \left ( \frac{f_\alpha}{2} + c \right ) r_\alpha(f) + I \left ( \alpha + r_\alpha(f) \right ) 
\end{align}
Because $V_\alpha(f_\alpha, f_{-\alpha})$ is concave in $f_\alpha$ for any $f_{-\alpha}$, a unique equilibrium exists~\citep{rosen1965existence}.
Because attention utility is a concave increasing function, lower-ability users have stronger incentives to increase their attention utility, and will tolerate larger consumption utility decreases in equilibrium.
As such, $\alpha > \alpha'$ implies $f^*_\alpha < f^*_{\alpha'}$.

\subsection{Illustration of market equilibrium }

Figure~\ref{fig:individual_outcomes_unif_het} illustrates the market equilibrium with and without attention bartering, in a platform where users are characterized by heterogeneous consumption preferences (as described in the beginning of this section).
Our goal is to highlight the differences between the effects of attention bartering in the homogeneous and the heterogeneous consumption preferences case.
\begin{figure}[h!]
\begin{center} 
\caption{Following, consumption, and attention with heterogeneous user preferences.
\label{fig:individual_outcomes_unif_het}}
\begin{subfigure}[t]{1 \textwidth}
\caption{Hidden follower information.}
\label{fig:no_info_het_uniform}
	\begin{subfigure}[t]{0.5 \textwidth}
  		\includegraphics[width = \linewidth]{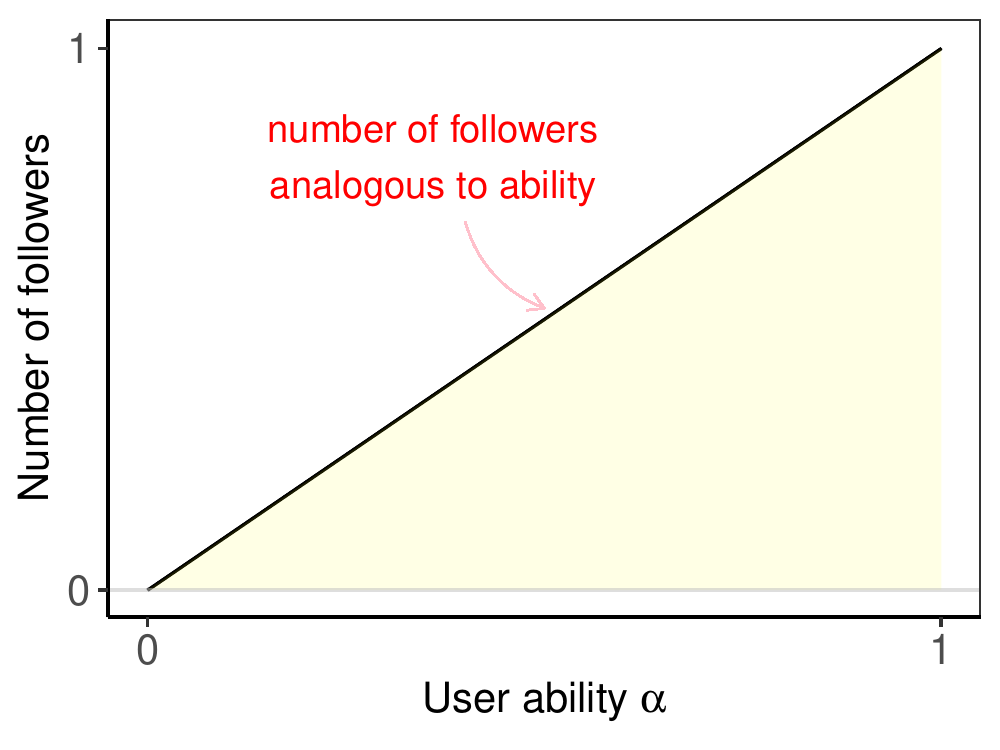}
		\label{fig:het_no_info_uniform_followers}
	\end{subfigure}
	\begin{subfigure}[t]{0.5 \textwidth}
  		\includegraphics[width = \linewidth]{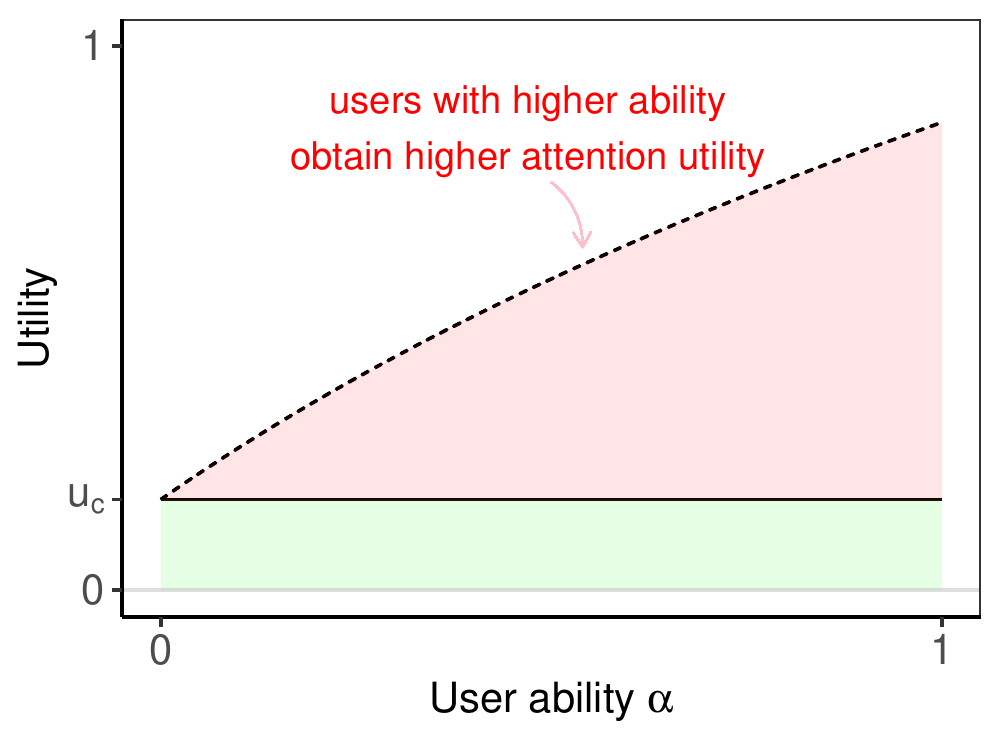}
  		\label{fig:het_no_info_uniform_utilities} 
	\end{subfigure}
\end{subfigure}

\begin{subfigure}[t]{1 \textwidth}
\caption{Public follower information.}
\label{fig:het_info_uniform}
	\begin{subfigure}[t]{0.5 \textwidth}
  		\includegraphics[width = \linewidth]{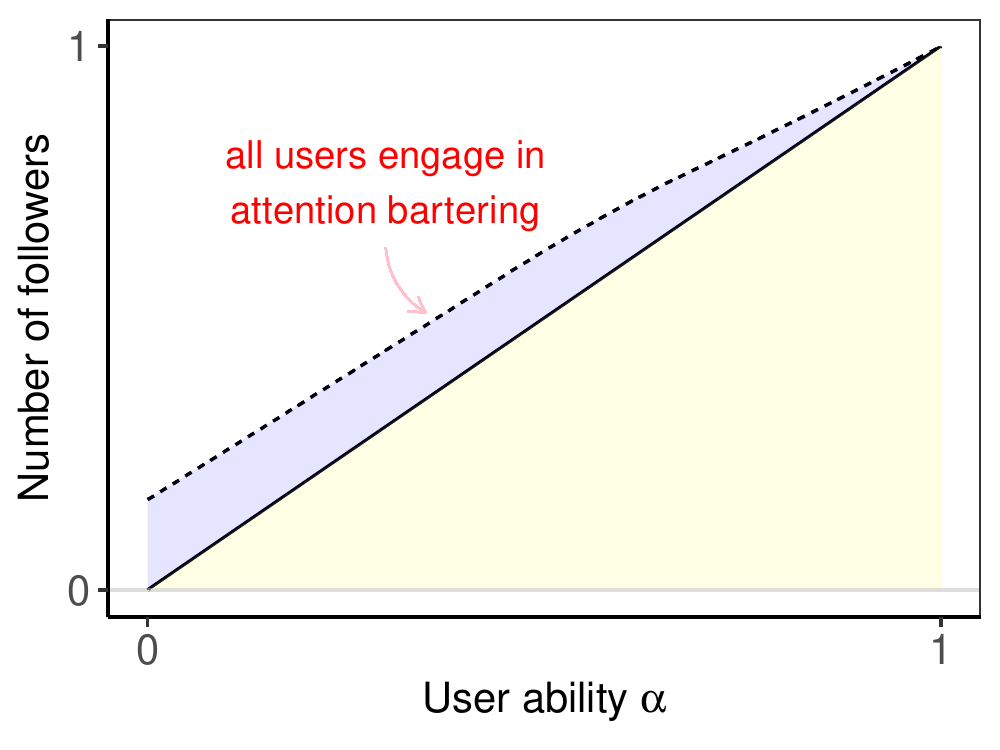}
		\label{fig:het_info_uniform_followers}
	\end{subfigure}
	\begin{subfigure}[t]{0.5 \textwidth}
  		\includegraphics[width = \linewidth]{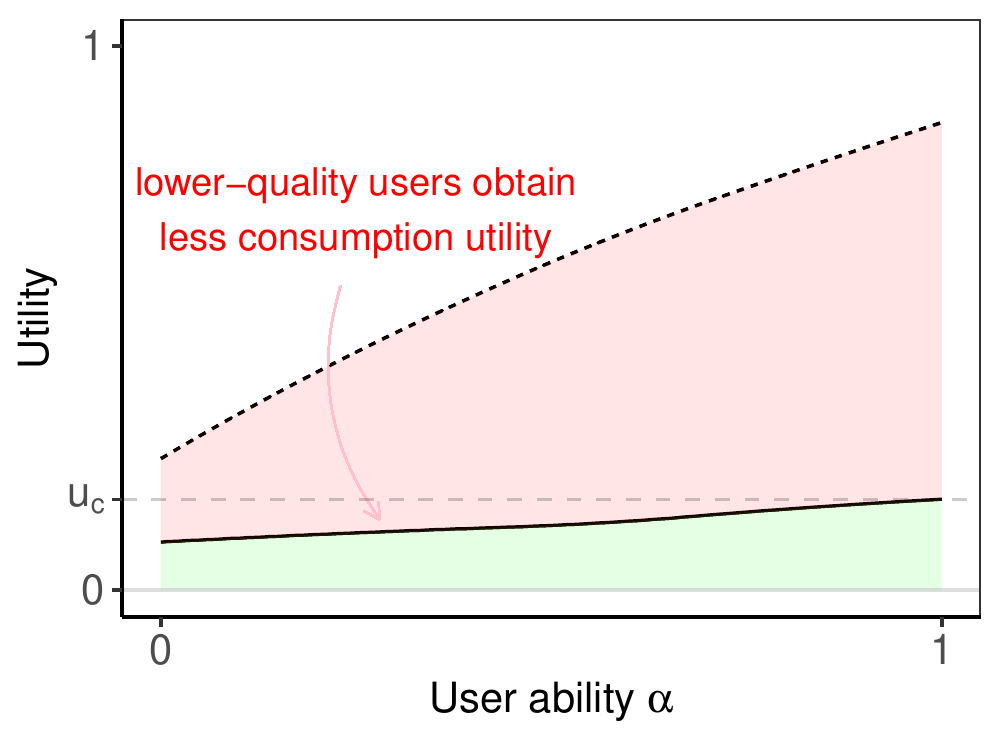}
  		\label{fig:het_info_uniform_utilities} 
	\end{subfigure}
\end{subfigure}
\end{center} 
{\footnotesize \begin{minipage}{1 \linewidth}
\emph{Notes:}
This figure plots users' followers, consumption utilities, and attention utilities, with and without follower information.
It illustrates the case of uniformly distributed valuations on the unit interval and $I(x)=\log(1+x)$ and $c=0.2$.
In all panels, the yellow-shaded area depicts the number of organic followers, 
the blue-shaded area depicts the number of reciprocal followers,
the green shaded-area depicts consumption utility, 
and the red shaded area depicts attention utility.
Panel~\ref{fig:no_info_het_uniform} plots the users' followers (left) and utilities (right) in a market with hidden follower information.
Only organic followers exist, and users are followed by a number of followers commensurate to their ability.
Every user obtains both consumption and attention utility: all users obtain the same consumption utility, but attention utility is commensurate to their abilities.
Panel~\ref{fig:het_info_uniform} plots the users' followers (left) and utilities (right) in a market with public follower information.
Reciprocal following now emerges, and users trade off decreases in consumption utility for increases in attention utility.
Figure~\ref{fig:statistics} reports additional quantities for the attention bartering equilibrium.
\end{minipage} }
\end{figure}

In Figure~\ref{fig:no_info_het_uniform}, we illustrate the equilibrium when attention bartering is not possible, for the case of $I(x) = log(1+x)$ and $c=0.2$.
Introducing heterogeneous consumption preferences ``smooths out'' following relationships, as each user attracts a followers in accordance to her ability.
This is in contrast to the homogeneous preferences case, where the number of followers is a step function (see Figure~\ref{fig:no_info_uniform}).

The equilibrium when attention bartering is possible is depicted in Figure~\ref{fig:het_info_uniform}.
In this simulation, we set the attention utility to $I(x) = \log(1+x)$, and the monitoring cost to $c=0.2$.
In the left panel, we observe that users find it utility-improving to engage in reciprocal following---except users of the highest ability, $\alpha=1$, who are already followed by every user of the platform. 
In the right panel, we plot the corresponding utilities.
All users obtain lower consumption utility, because they follow users of abilities less than their opportunity cost, and they incur monitoring costs to maintain the reciprocal relationships.
Lower-ability users engage in reciprocation most intensely: these users reciprocally follow higher-ability users, and hence obtain less consumption utility, but also the highest attention utility increase.

\begin{figure}[h!]
\caption{Followers, followees, and ratios when attention bartering is possible}
\label{fig:statistics} 
\begin{center}
\begin{subfigure}[t]{0.49 \textwidth}
\caption{Homogeneous consumption preferences}
\includegraphics[width = \linewidth]{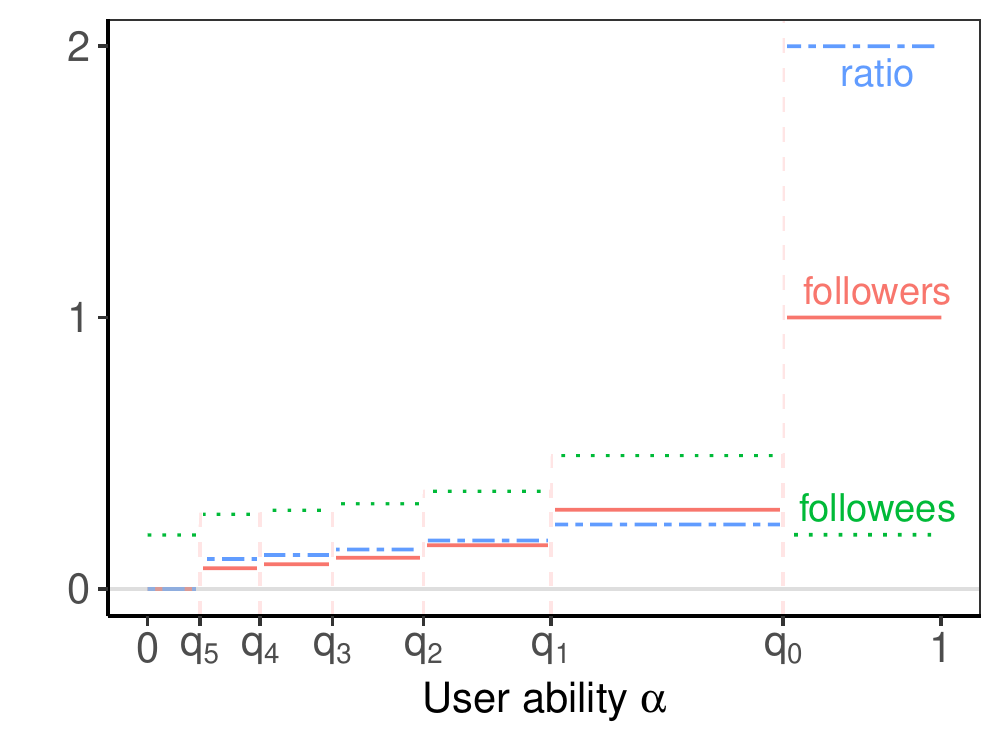}
\label{fig:info_uniform_ratio} 
\end{subfigure}	
\begin{subfigure}[t]{0.49 \textwidth}
\caption{Heterogeneous consumption  preferences}
\includegraphics[width = \linewidth]{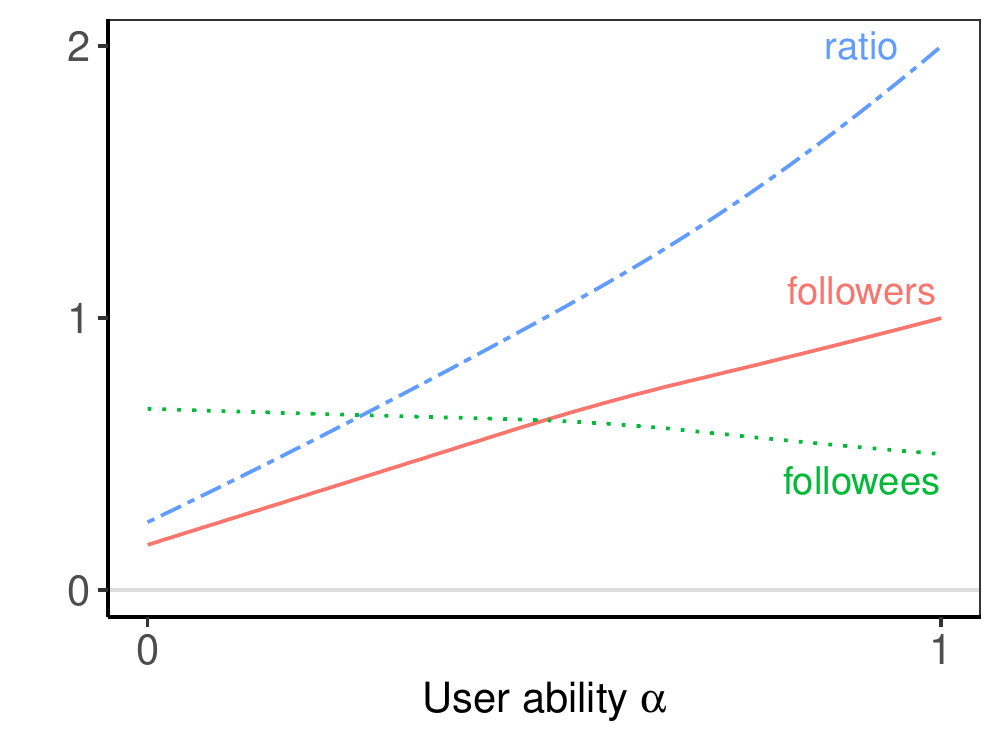}
\label{fig:het_info_uniform_ratio} 
\end{subfigure}
\end{center}
{\footnotesize \begin{minipage}{1 \linewidth}
\vspace*{-0.2in}
\emph{Notes:}
This figure plots network statistics in the equilibrium of a platform where  attention bartering is possible.
Panel~\ref{fig:info_uniform_ratio} shows the number of followees (dotted green line), 
the number of followers (solid red line), 
and the follower-to-followee ratio (two-dash blue line)  for each user ability $\alpha$, when consumption preferences are homogeneous, and with all other parameters set as in in Figure~\ref{fig:individual_outcomes_unif}.
Panel~\ref{fig:het_info_uniform_ratio} shows the same quantities when consumption preferences are heterogeneous, and with all other parameters set as in in Figure~\ref{fig:individual_outcomes_unif_het}.
\end{minipage} }
\end{figure}

\newpage \clearpage
\section{Additional empirical results} 
\label{app:empirics}

\subsection{Follower, followee, and ratio distributions of all users}
Figure~\ref{fig:alltwitter_stats_distr} reports the distributions of followers, followees, and follower-to-followee ratios of users in the entire sample.
For each outcome, we report a Gaussian kernel density estimate of its distribution with the bandwidth selected using  Silverman's rule of thumb~\citep{silverman1986density}, as well as the estimated median (solid red vertical line) and the mean (dashed red vertical line).

Compared to the average \ET{} user, Twitter users in our data have on average fewer followers and followees, and lower ratios: the median user in the entire data set has \ALLmedianFollowers{} followers, \ALLmedianFollowing{} followees, and a ratio equal to \ALLmedianRatio.
Only \ALLmorethanoneRatio{}\% of the users have ratio higher than one in the entire sample, and the distribution of ratios exhibits a noticeable kink at the unit ratio---this is consistent with  Proposition~\ref{prop:attention_bartering_properties_homo}~\ref{pprop:eqbm_followers_ratio}.
It is worth noting that the effects of Twitter’s followee limits are easily discernible in the kinks of the middle facet.

\begin{figure}[h!]
\begin{center}
\begin{minipage}{1 \linewidth}	
\caption{Distributions of followers, followees, and ratios of all users}
\label{fig:alltwitter_stats_distr}
\includegraphics[width = 1 \textwidth]{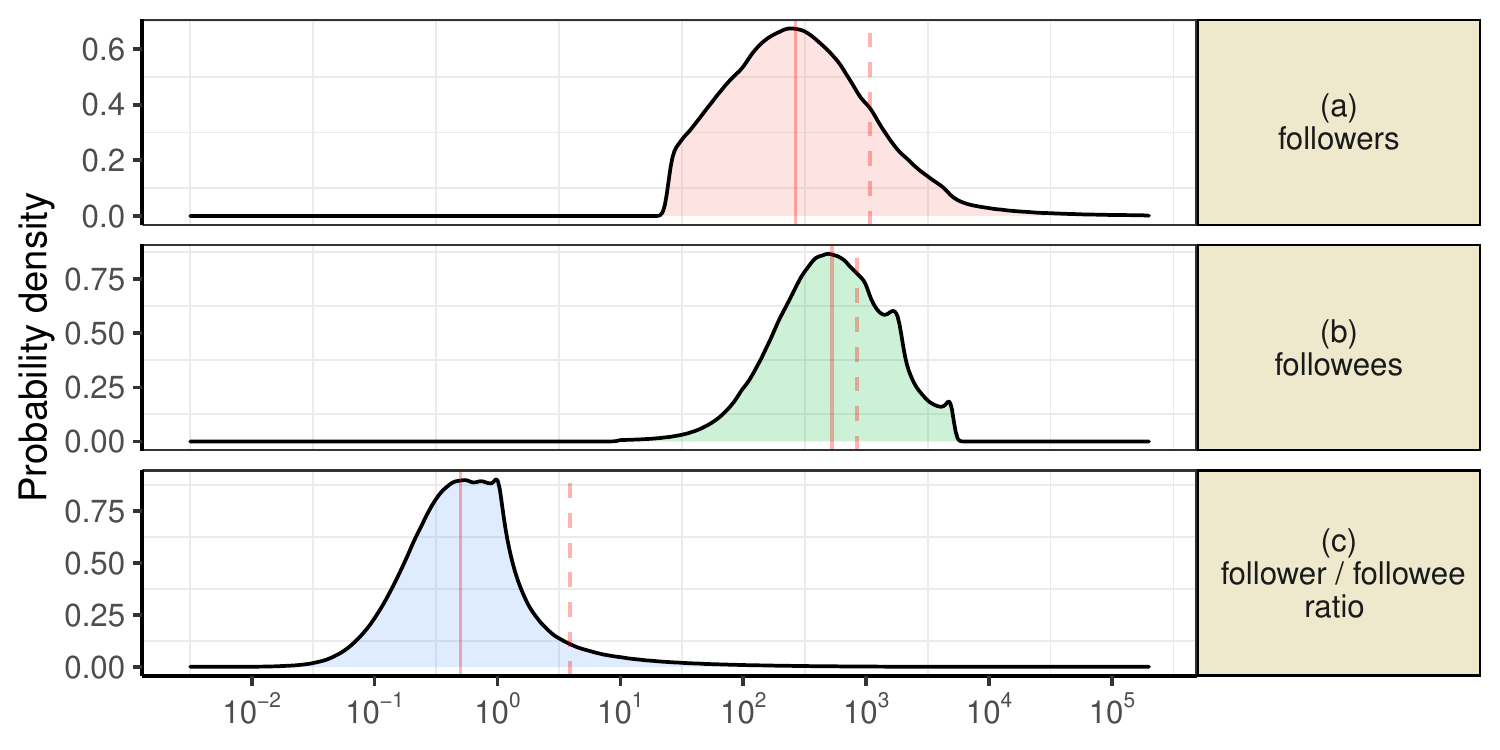}
\end{minipage}
\end{center}
{\footnotesize \begin{minipage}{1 \linewidth}
\vspace*{-0.0in}
\emph{Notes:}
This figure reports kernel density estimates of the probability density function of the distributions of followers, followees, and follower-to-followee ratios in the entire data.
For each facet, the vertical red lines depict the median (solid) and the mean (dashed) of the corresponding distribution, and the kernel bandwidth is selected using  Silverman's rule of thumb~\citep{silverman1986density}.
See Figure~\ref{fig:econtwitter_stats_distr} for density estimates of the same statistics for \ET{} users.
\end{minipage}}
\end{figure}

\newpage \clearpage 
\subsection{More details on the empirical estimates of Figure~\ref{fig:econtwitter_vs_theory_statistics}}

\begin{figure}[h!]
\begin{center}
\begin{minipage}{1 \linewidth}	
\caption{\ET{} estimates of network statistics. \label{fig:econtwitter_statistics_unscaled}}
\includegraphics[width = 0.98 \textwidth]{./plots/bigbadaplot_unscaled.pdf}	
\end{minipage}
\end{center}
{\footnotesize \begin{minipage}{1 \linewidth}
\vspace*{-0.1in}
\emph{Notes:}
This figure reports unscaled \ET{} estimates of various network statistics, as a function of user ability.
For more details on the construction of the estimates, see the description of Figure~\ref{fig:econtwitter_vs_theory_statistics}.
\end{minipage}}
\end{figure}

\newpage \clearpage
\subsection{Ability is not highly correlated with tenure}
One alternative explanation for the patterns we observe in the \ET{} data can be found with preferential attachment.
Namely, if more able users join the platform earlier, then these users would attract more followers---ability is not the reason for attracting followers, but time of entry is. 
Furthermore, these users may also follow more able users on average because those users were active at the time of relationship formation. 

Table~\ref{tab:age_ability} examines the relation of tenure, ability, and users' follower-to-followee ratios. 
Tenure is positively correlated with ability and with follower-to-followee ratios, but it only explains about 4\% of the variability in both quantities.
In contrast, user ability is about eight times more predictive of a users' follower-to-followee ratio.
The lack of predictive power of user tenure suggests that the preferential attachment explanation is not strongly supported in our data. 

\begin{table}[h!] \centering 
  \caption{The relation of tenure, ability, and follower-to-followee ratio.} 
  \label{tab:age_ability} 
\begin{tabular}{@{\extracolsep{5pt}}lccc} 
\\[-1.8ex]\hline 
\hline \\[-1.8ex] 
 & \multicolumn{3}{c}{\textit{Dependent variable:}} \\ 
\cline{2-4} 
\\[-1.8ex] & Ability (percentile) & \multicolumn{2}{c}{Follower/followee Ratio} \\ 
\\[-1.8ex] & (1) & (2) & (3)\\ 
\hline \\[-1.8ex] 
 Tenure (days) & 0.006$^{***}$ & 0.0003$^{***}$ &  \\ 
  & (0.0001) & (0.00001) &  \\ 
  & & & \\ 
 Ability (percentile) &  &  & 0.025$^{***}$ \\ 
  &  &  & (0.0001) \\ 
  & & & \\ 
 Constant & 33.461$^{***}$ & 0.777$^{***}$ & 0.267$^{***}$ \\ 
  & (0.375) & (0.016) & (0.008) \\ 
  & & & \\ 
\hline \\[-1.8ex] 
Observations & 55,231 & 55,231 & 55,231 \\ 
R$^{2}$ & 0.041 & 0.042 & 0.344 \\ 
Adjusted R$^{2}$ & 0.041 & 0.042 & 0.344 \\ 
Residual Std. Error (df = 55229) & 28.271 & 1.184 & 0.980 \\ 
F Statistic (df = 1; 55229) & 2,338.288$^{***}$ & 2,436.529$^{***}$ & 28,907.850$^{***}$ \\ 
\hline 
\hline \\[-1.8ex] 
\end{tabular}
\\
\begin{minipage}{1\textwidth}
\begin{footnotesize}
\begin{onehalfspacing}
\emph{Notes:} This table reports regressions where the dependent variable is ability (columne 1) and follower-to-followee ratio (columns 2 and 3).
                                        The independent variables are users' tenure on the platform, and users' abilities.
                                        For more details on the definition of user ability, see Section~\ref{ssec:user_ability}
\starlanguage
\end{onehalfspacing}
\end{footnotesize}
\end{minipage}
\end{table}

\end{document}